\documentclass[journal,twoside,web]{ieeecolor}
\usepackage{generic}
\usepackage{cite}
\usepackage{amsmath,amssymb,amsfonts}
\usepackage{algorithmic}
\usepackage{graphicx}
\usepackage{textcomp}
\def\BibTeX{{\rm B\kern-.05em{\sc i\kern-.025em b}\kern-.08em
    T\kern-.1667em\lower.7ex\hbox{E}\kern-.125emX}}
\newtheorem{assumption}{Assumption}
\newtheorem{lemma}{Lemma}
\newtheorem{theorem}{Theorem}

\newtheorem{remark}{Remark}
\newtheorem{definition}{Definition}
\def\dref#1{(\ref{#1})}

\begin{document}
\title{An Operator-Theoretic Approach to Robust Event-Triggered Control of Network Systems with Frequency-Domain Uncertainties}
\author{Shiqi Zhang, Yuezu Lv, \IEEEmembership{Member, IEEE}, and Zhongkui Li,  \IEEEmembership{Member, IEEE}
\thanks{This work was supported by the National Natural Science Foundation of China under grants 61973006 and U1713223. }
\thanks{S. Zhang and Z. Li are with the State Key Laboratory for Turbulence and Complex Systems, Department of Mechanics and Engineering Science, College of Engineering, Peking University, Beijing 100871, China (e-mail: zsqpkuedu@pku.edu.cn;zhongkli@pku.edu.cn). }
\thanks{Y. Lv is with the Department of Systems Science, School of Mathematics, Southeast University, Nanjing 211189, China (email: yzlv@seu.edu.cn)   }
\thanks{ }}

\maketitle

\begin{abstract}
In this paper, we study the robustness of the event-triggered consensus algorithms against frequency-domain uncertainties. It is revealed that the sampling errors resulted by event triggering are essentially  images of linear finite-gain $\mathcal L_2$-stable operators acting on the consensus errors of the sampled states and 	 the event-triggered  mechanism is  equivalent to a negative feedback loop introduced additionally to the feedback system. In virtue of this,  the robust consensus problem of the event-triggered network systems subject to additive dynamic uncertainties and network multiplicative uncertainties are considered, respectively. 
In both cases, quantitative relationships among the parameters of the controllers, the Laplacian matrix of the network topology, and the robustness against aperiodic event triggering and frequency-domain uncertainties are unveiled. 
 Furthermore, the event-triggered dynamic average consensus (DAC) problem is also investigated, wherein the sampling errors are shown to be images of nonlinear finite-gain operators. The robust performance of the proposed DAC algorithm is analyzed, which indicates that the robustness and the performance are negatively related to the eigenratio of the Laplacian matrix. Simulation examples are also provided to verify the obtained results.
\end{abstract}


\section{Introduction}

\label{sec:introduction}
\IEEEPARstart{E}{v}ent-triggered control,  whose basic idea  is to replace the continuous or periodic sampling mechanism by the  aperiodic  and sporadic one in the control algorithms \cite{Arzen1999eventbased}, \cite{Astrom2002comparison},
\cite{muller2009}, \cite{Tabuada2007eventtrigger},  \cite{heemels2012introduction},
\cite{Miskowicz2015eventbased},
 originates form the aperiodic sampling problem \cite{hetel2017survey} and shows great effectivity  when applied in controlling continuous-time systems with digital controllers.   
In the last decade or so, it has been further introduced to distributed control  of network systems \cite{Lemmon2008eventtriggered}, \cite{Lemmon2011eventtriggering}, \cite{Dimos2012distributed}, \cite{Tan2019consensus}, \cite{Cao2021distributed}, \cite{Wen2016}, \cite{Xu2019}, wherein not only the sampling mechanism but also the communicating mechanism among agents is event-based.
 
 Compared to the continuous or time-driven distributed control algorithms, event-triggered ones release information-exchange burden and thus have lower communication cost. Moreover, for some environments in which  continuous communication is restricted, forbidden or impossible, the event-triggered algorithms are more practical  \cite{Cortes2019survey}. Typical works on distributed event-triggered control include \cite{cortes2016distributed},\cite{dimarogonas2012distributed}, \cite{Garcia2013decentralized}, \cite{berneburg2019ditributed} for  integrator networks and \cite{Dimer2012a}, \cite{Dimer2012b},
\cite{Garcia2014decentralized}
 for general linear systems. Reference \cite{kia2015distributed} introduced the event-triggered mechanism to dynamic average consensus (DAC) algorithms. \cite{Tan2019consensus}, \cite{Wen2016}, \cite{Xu2019}, \cite{Garcia2018}, \cite{Dimos2015event-triggered},\cite{cheng2016leaderfollowing} considered the event-triggered leader-following tracking problems and  \cite{Cao2021distributed} further considered distributed dynamic event-triggered control for nonlinear multi-agent systems.   

In the aforementioned works,  the effectivity of the event-triggered control algorithms relies on the assumption that the  system dynamics  are accurately known. This assumption, however, is too stringent in reality due to the  ubiquitous unmodeled dynamics, the omnipresent communication constraints and the universal parametric uncertainties. Therefore, it is  quite an imperative task to examine the robustness of  the event-triggered control algorithms in the presence of  uncertainties. To the knowledge of the authors, there are few works along this line, except \cite{Paul2017stabilization,Seuret2019robust}.        
 Reference \cite{Paul2017stabilization} considered the robust event-triggered stabilization problem for  discrete-time systems and \cite{Seuret2019robust} studied the continuous-time case. These works are fairly important in the sense that they provide  conditions under which the  event-triggered control algorithms can still work in the presence of time-domain uncertainties. 

Apart from time-domain parametric uncertainties, frequency-domain uncertainties, including  unmodeled dynamics, and modeling errors,  are  a more general class of uncertainties \cite{zhou1998essentials}, which may be encountered and need to be dealt with in the event-triggered control problem. In  network systems, communication delays, package dropping and  network uncertainties  are also very prevalent phenomena and thus put forward new challenges to the distributed event-triggered control algorithms.
In virtue of these observations, in this paper, we intend to handle  the robustness of  distributed event-triggered control algorithms against frequency-domain uncertainties.              
 
Distributed control  algorithms with continuous and ideal communications  have been proved robust to various kinds of frequency-domain uncertainties such as additive dynamic uncertainties  \cite{trentelman2013robust}, \cite{li2018robust},
network multiplicative uncertainties  \cite{lirobust2017}, \cite{zelazo2017robustness},\cite{lirobust2019},  coprime factor uncertainties \cite{trentelman2016robust} and so on. It is a natural question that whether  distributed event-triggered algorithms also possess the robustness against frequency-domain uncertainties.    
To answer this question, it is  necessary to  adopt the frequency-domain robust control tools such as the small gain theorem and the  $\mu$ analysis.
However, essentially speaking, the event-triggered algorithms belong to a  special branch of the aperiodic sampling algorithms, whose definition, modeling, and methodology are all based on the time-domain analysis. More specifically, the triggering function, which decides  whether the certain agent updates its state estimation and broadcasts it to its  neighbors, is expressed in  a time-domain form. Moreover, it characterizes a time-domain point-wise inequality constraint of the sampling  error.
Therefore, the analysis and design of the event-triggered control problem have almost always been based on the  Lyapunov stability analysis, which is severely different from the robust analysis and synthesis tools mentioned above. 
In a word, because of the systematic gap between the time-domain event triggering  and the frequency-domain uncertainties, the robustness of event-triggered control of network systems against frequency-domain uncertainties still remains an open and challenging problem. 

To solve this problem, one of the  main difficulties is how to build a bridge between the time-domain sampling mechanism and  frequency-domain uncertainties, or in other words, how to `translate' the sampling mechanism into the  frequency-domain language. 
In this paper, we utilize the operator theory as our basic tool to unify them. 
By studying  the frequency-domain properties of the sampling  errors, we find that in classical event-triggered consensus algorithms, the sampling errors are images of the linear  operators acting on the consensus error of the sampled states. 
It is worth noting that these linear operators are neither generally rational transfer matrices in $\mathcal{RH}_{\infty}$ nor sector bound uncertainties of logarithmic quantizers as in \cite{xie}.    
Nevertheless, we can ascertain from the triggering function that the operators are finite gain $\mathcal L_2 $ stable. The operator  gain, depending on the sampling parameters and the Laplacian matrix of the topology graph, characterizes the extent of the sampling error introduced by event triggering. In light of this, the event-triggered mechanism is  equivalent to additionally introducing a negative feedback loop consisting of the linear operators. 

One of the crucial advantages of the proposed operator-theoretic approach is that we can handle the robustness of the event-triggered consensus problem of network systems with respect to various kinds of frequency-domain uncertainties. In this paper,  we consider additive dynamic uncertainties and network topology uncertainties for illustration. 
Inherent constraints on the robustness, imposed by the parameters of the controllers, the network topology, the bounds of the additive/multiplicative uncertainties, and the gains of the operators representing the event-triggered sampling,  are unveiled. Moreover, the results can be extended to the event-triggered dynamic average consensus (DAC) problem, where the sampling errors are found to be images of nonlinear but finite-gain operators acting on the consensus errors of the sampled states. Especially, we   consider  the event-triggered robust DAC problem and examine the performance of the proposed DAC algorithm under additive dynamic uncertainties. It is shown that the smaller the eigenratio is, the better the robustness and the tracking performance will be  under event triggering  and additive dynamic uncertainties.  

The remaining part of this paper is organized   as follows:
In Section \ref{s2}, we introduce some necessary mathematical preliminaries. In Section \ref{s3}, we revisit the event-triggered consensus algorithm from a  an operator-theoretic perspective. In Section \ref{s4}, we study the robustness of the event-triggered consensus algorithms against frequency-domain uncertainties under the operator-theoretic framework.  
In Section \ref{s5}, we extend the results to the event-triggered DAC problem. Section \ref{s6} provides some simulation examples for illustration and Section \ref{conclusion} concludes this paper. 
      
{ {\it Notations: } The notations used in this paper are fairly standard. $\mathbf R^{n\times m}$ denotes the linear space of all $n\times m$-dimensional matrices. $I_N$ represents the $N$-dimensional identity matrix and $1_N$ represents the $N$-dimensional vector whose elements are all equal to $1$.  $\mathrm{diag}\{a_1,\cdots,a_n\}$ denotes the diagonal matrix whose diagonal elements are equal to $a_1,\cdots,a_n$. The set of all real rational stable transfer matrices is denoted by $\mathcal{RH}_{\infty}$.  }

\section{Mathematical Preliminaries }\label{s2}
This section reviews some useful results and conclusions from the operator  theory in Subsection \ref{A}, from the robust control theory in Subsection \ref{B} and from the graph theory in Subsection \ref{C}, respectively.

\subsection{Operator Theory and Hilbert  Space }\label{A}
 
\begin{definition}\label{df1}
	\cite{Desoer1975feedback} Let $\mathcal H_1$ and $\mathcal H_2$ be two Banach  spaces. An operator $\phi(\cdot):\mathcal H_1 \mapsto \mathcal H_2$ is called a linear operator if the following two conditions hold:
	
	1)  $\phi(x+y)=\phi(x)+\phi(y)$, for $\forall x,y \in \mathcal H_1$;
	
	2) $\phi(\lambda x)=\lambda \phi(x),$ for $\forall$ $\lambda\in \mathbf R$.
\end{definition}
	

\begin{definition}
	The $\mathcal L_2$ norm  of a signal $f$  is defined as $$\|f\|_2=\sqrt{\int_0^{\infty}f^*(t)f(t)dt}=\sqrt{\frac{1}{2\pi}\int_{-\infty}^{\infty}f^*(j\omega)f(j\omega)d\omega},$$ where $\mathcal L_2$ denotes the Banach space with $\mathcal L_2$ norm well defined. 
\end{definition}
\begin{definition}
	\cite{Desoer1975feedback} Letting $\phi(\cdot):\mathcal L_2\mapsto\mathcal L_2$ be an operator  such that for $\forall x\in \mathcal L_2$, $\|\phi(x)\|_2\leq\gamma\|x\|_2+\beta$, where $\gamma>0$ and $\beta>0$ are positive constants, then this operator is called a finite-gain operator with operator norm $\|\phi\|_{\infty}\leq\gamma$.
\end{definition}
\begin{lemma}\label{l2}
	\cite{Desoer1975feedback} Let $\phi(\cdot):\mathcal{L}_2\mapsto \mathcal{L}_2$ denote a linear  (finite-gain)  operator in the time domain and suppose that $y(t)=\phi(x(t))$, where $y(t)$ and $x(t)$ are vectors in the  $\mathcal L_2$ space. Denote by $y(s)$ and $x(s)$ the Laplace transformation of $y(t)$ and $x(t)$. It then follows that $y(s)=\Delta(x(s))$, where $  \Delta(\cdot):\mathcal L_2\mapsto \mathcal L_2$ denotes a linear  (finite-gain) operator in the frequency domain. 
	\end{lemma}


\subsection{Robust Control Theory}\label{B}

\begin{lemma}\label{smallgain}
	\cite{Desoer1975feedback} (Small Gain Theorem) Supposing that $G(\cdot), \Delta(\cdot):\mathcal L_2\mapsto\mathcal L_2$ are finite-gain operators with operator norms $\|G\|_{\infty}=\gamma_1$ and $\|\Delta\|_{\infty}=\gamma_2$, the system interconnection shown in Fig. \ref{loop} is internally stable, if $\gamma_1\gamma_2<1$.
 \end{lemma}
\begin{definition}[\cite{zhou1998essentials}]\label{added} { 
	Let $\mathbf\Delta$ represent the set of structured finite-gain $\mathcal L_2$-{ stable }operators. For $M\in \mathbf R^{m\times n}$, $\mu_{\mathbf\Delta}(M)$ is defined as 
	$$\mu_{\mathbf\Delta}(M)=\frac{1}{\min \{\bar\sigma(\Delta):\Delta\in \mathbf \Delta, \det(I-M\Delta)=0 \}}$$
	unless no $\Delta\in \mathbf \Delta$ makes $\det(I-M\Delta)$ singular, in which case $\mu_{\mathbf\Delta}(M):=0$.}
	
\end{definition}

\begin{lemma}[\cite{zhou1998essentials}]\label{l3}
Let $\mathbf\Delta$ represent the set of structured finite-gain $\mathcal L_2$-{ stable }operators. The loop shown in Fig. \ref{loop} is well-posed and internally stable for all $\Delta\in \mathbf\Delta$ with operator norm $\|\Delta\|_{\infty}\leq\gamma$ if and only if  
$
\sup_{\omega\in \mathbf{R}}\mu_{\mathbf\Delta}(G(j\omega))<\frac{1}{\gamma},
$
where $\mu_{\mathbf\Delta}(\cdot)$ denotes the structured singular value. 
\begin{figure}[!t]
\centerline{\includegraphics[width=0.8\columnwidth]{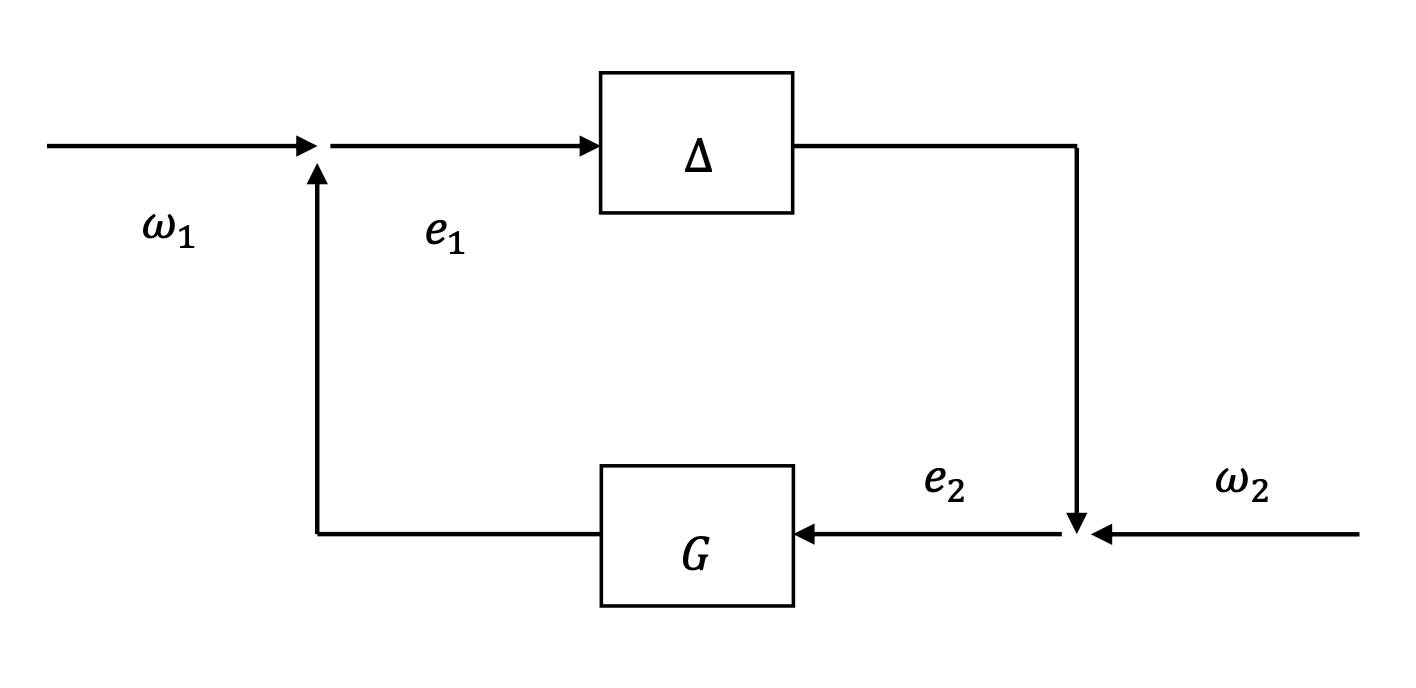}}
\caption{System interconnection.}
\label{loop}
\end{figure}
\end{lemma}

\begin{lemma}[\cite{zhou1998essentials}]\label{l4}{ 
 Assume that $G(s)=\left[\begin{smallmatrix}
		G_{11}(s) & G_{12}(s)\\
		G_{21}(s) & G_{22}(s)
	\end{smallmatrix}\right ]$
	and $\mathbf\Delta=\left[\begin{smallmatrix}
		\mathbf\Delta_1& 0\\
		 0 & \mathbf\Delta_2
	\end{smallmatrix}\right]$ is the set of all finite-gain $\mathcal L_2$ stable block diagonal  operators with compatible dimensions with $G(s)$. Then we have $\mu_{\mathbf\Delta}(G(j\omega))\leq\sqrt{\|G_{11}(j\omega)\|^2+\|G_{22}(j\omega)\|^2+2\|G_{12}(j\omega)\|\|G_{21}(j\omega)\|}$, $\forall \omega\in \mathbf R$.
Moreover, suppose that ${ \mathbf\Delta}=\left[\begin{smallmatrix}
		\mathbf \Delta_1 & &\\
		& \mathbf \Delta_2 &\\
		& & \mathbf \Delta_3
	\end{smallmatrix}\right]$ and $G=\left[\begin{smallmatrix}
		G_{11} & G_{12} & G_{13}\\
		G_{21} & G_{22} & G_{23}\\
		G_{31} & G_{32} & G_{33}
	\end{smallmatrix}\right].$
	Then,
	$$\begin{aligned}
		&\mu_{\mathbf \Delta}^2(G(j\omega))\leq\|G_{11}(j\omega)\|^2+\|G_{22}(j\omega)\|^2+\|G_{33}(j\omega)\|^2\\
		& \quad+2\|G_{12}(j\omega)\|\|G_{21}(j\omega)\|+2\|G_{13}(j\omega)\|\|G_{31}(j\omega)\|\\
		&\quad+ 2\|G_{23}(j\omega)\|\|G_{32}(j\omega)\|.
	\end{aligned}
	$$}
\end{lemma}

\begin{lemma}[\cite{zhou1998essentials}]\label{l5}
	Let $\gamma>0$ and  $\mathcal F_u(M,\Delta)=M_{22}+M_{21}\Delta(I-M_{11}\Delta)^{-1}M_{12}$ denote the upper linear fractional transformation with respect to $\Delta$. For all $\Delta\in \mathbf \Delta$ with $\|\Delta\|_{\infty}\leq\gamma$, the transfer function $\mathcal F_u(G_p,\Delta)$ is internally stable  and $\|\mathcal F_u(G_p,\Delta)\|_{\infty}<\frac{1}{\gamma}$ if and only if 
	$
	\sup_{\omega\in \mathbf R} \mu_{\mathbf\Delta_p}(G_p(j\omega))<\frac{1}{\gamma},
	$
	where $\mathbf\Delta_p=\left[\begin{smallmatrix}
		\mathbf\Delta & \\
		 & \mathbf\Delta_f
	\end{smallmatrix}\right].$
\end{lemma}

\subsection{Graph Theory}\label{C}
An  undirected graph $\mathcal G\{\mathcal V,\mathcal E, \mathcal W\}$  describes the network topology among the agents, where $\mathcal V=\{1,2,\cdots, N\}$ denotes the set of vertices, $\mathcal E =\{1,2,\cdots,m\}$ denotes the set of the edges, and $\mathcal W=\{\cdots,w_{(i,j)},\cdots\}$ or equivalently $\{w_1,w_2,\cdots,w_m\}$ denotes the set of the weights corresponding to the edges.
The adjacency matrix of the graph $\mathcal G$ is denoted as $\mathcal A$ and $a_{ij}$ is the $(i,j)$-th element of $\mathcal A$ defined  as $a_{ij}=w_{ij}$ if $(i,j)\in \mathcal E$ and  $a_{ij}=0$ otherwise. Letting $d_i=\sum_{j=1}^Na_{ij}$ be the degree of the node $i$  and $H=\mathrm{diag}\{d_1,\cdots,d_N\}$. The Laplacian matrix of the graph $\mathcal G$ is then defined as $L=H-A$. Let $D\in \mathbf R^{N\times m}$ be the incidence matrix of $\mathcal G$ such that $d_{ij}=-1$ if $i$ is the tail of the edge $(i,j)$ and $d_{ij}=1$ if $i$ is the head of the edge $(i,j)$ and $d_{ij}=0$ otherwise. It is easy to see that the sum of  each column of $D$ is equal to $0$. 
\begin{lemma}
	\cite{Mesbahi} For an undirected graph $\mathcal G$, the Laplacian matrix $L=DWD^T$, where $W=\mathrm{diag}\{w_1,w_2,\cdots,w_{m}\}$.
\end{lemma}

\section{ Revisit of the Event-Triggered Algorithm From a Robust Control Perspective}\label{s3}
Consider a network  consisting of $N$ single-input-single-output agents with scalar states. The dynamics of each agent can be described by a single integrator: 
\begin{equation}\label{eq1}
\begin{aligned}
	&\dot x_i = u_i, \quad  i=1,\cdots,N,\\
\end{aligned}
\end{equation} with $x_i$ as the state variable of agent $i$ and $u_i$  the control input. 
It is assumed that the communication among the agents is depicted by a graph $\mathcal{G}$. The control objective of distributed  algorithms is to ensure that the states of the agents reach consensus, i.e., $x_i-x_j\rightarrow0$, $\forall i,j\in \{1,\cdots,N\} $ as $t\rightarrow \infty$.
 Throughout this paper, the following assumption holds.
\begin{assumption}
	The communication graph $\mathcal{G}$ is undirected and connected.
\end{assumption}
 
 In this paper, we consider the event-triggered mechanism. Under this mechanism, instead of  continuous  communication between agents, each agent (say agent $i$) only updates the estimate of its state ($\hat x_i$) to the real value of $x_i$ and send it to its neighbors at the triggering instants, between which the estimate is calculated locally by itself and its neighbors.    
 We set the initial time $t_0$ as the first triggering instant of each agent and define a triggering function:\begin{equation} \label{triggerfunction1}
	f_i=e_i^2-\alpha \sum_{j=1}^Na_{ij}(\hat x_i-\hat x_j)^2-\mu e^{-\nu t},
\end{equation} where $\alpha,\mu,\nu$ are positive constants and $e_i$ denotes the gap between the estimate $\hat x_i$ and the real state $x_i$, i.e., $e_i=\hat x_i-x_i$.  The next event will happen whenever the triggering condition $f_i\geq 0$ is satisfied, i.e., $t_{k+1}^i=\inf \{t|t>t_{k}^i ,f_i\geq 0\}$.
 
In this section, we consider the following distributed control law:
\begin{equation}\label{eq2}
	u_i =-\beta\sum_{j=1}^Na_{ij}(\hat x_i-\hat x_j), \quad i=1,\cdots,N.
\end{equation}
During the time interval between two triggering instants, the estimate $\hat x_i$ used by all its neighbors is held to be a constant, i.e., $\hat x_i(t)=x_i(t_{k}^i)$, $\forall t\in [t_{k}^i,t_{k+1}^i)$. This estimating mechanism is called zero-order holder (ZOH) \cite{Cortes2019survey}.

The closed-loop system derived from \dref{eq1} and \dref{eq2} is 
\begin{equation}\label{eq3}
	\dot x_i=-\beta\sum_{j=1}^Na_{ij}(\hat x_i-\hat x_j), \quad i = 1,\cdots,N.
\end{equation}
Define $z_i=\frac{1}{N}\sum_{j=1}^N(\hat x_i-\hat x_j)$. Then it follows that 
\begin{equation}\label{eq4}
\begin{aligned}
	&\dot x_i = -\beta\sum_{j=1}^Na_{ij}\left[(x_i- x_j)+(e_i- e_j)\right],\\
	& z_i=\frac{1}{N}\sum_{j=1}^N[( x_i- x_j)+( e_i- e_j)].
\end{aligned}
\end{equation}
We can further  rewrite \dref{eq4} in a compact form as
\begin{equation}\label{closeloooooooop}
	\begin{aligned}
		&\dot x=-\beta Lx-\beta L e,\\
		&z=Mx+Me,
	\end{aligned}
\end{equation}
where $M=I-\frac{1}{N}1_N1_N^T$ and $L$ is the Laplacian matrix of the graph $\mathcal{G}$.

	It is clear from the triggering mechanism and the triggering function \dref{triggerfunction1} that $e_i$ is reset to be zero at each triggering instant and increases from $0$ to some positive value during two triggering instants, and then drops again to zero at the next triggering instant. The square of the sampling error $e_i$ is bounded from above by a quadratic form of $z(t)$ and an exponential decaying term at any time instant. While this bound relationship is described by the time domain terminology,  we discover   a frequency-domain relationship between $e(s)$ and $z(s)$, as will be unveiled in the next theorem.

\begin{theorem}\label{th1}
	For the closed-loop network  \dref{closeloooooooop} with the triggering function  \dref{triggerfunction1}, it  follows that $e(s)=\Delta ( z(s))$, where $\Delta(\cdot): \mathcal{L}_2 \mapsto \mathcal{L}_2$ is a linear operator in the frequency domain  with $\|\Delta(z)\|_2\leq \sqrt{2\alpha \lambda_N}\|z\|_2+\sqrt{\frac{N\mu}{\nu}},$ i.e., $\|\Delta\|_{\infty}\leq\sqrt{2\alpha\lambda_N}$, where $\lambda_N $ denotes the largest eigenvalue of $L$.  
\end{theorem}
\begin{proof}
	From \dref{eq3} and the ZOH  mechanism of $\hat x_i$, it is not difficult to obtain that 
	$$
	\begin{aligned}
	 e_i(t)&=\hat x_i(t)-x_i(t)\\
	 	 & =\beta\sum_{j=1}^Na_{ij}\int_{t_{k}^i}^t(\hat x_i(\tau)-\hat x_j(\tau))d\tau\\
	 & = \beta L_i^T\int_{t_{k}^i}^t\hat x(\tau)d\tau \\
	 & = \beta L_i^TM\int_{t_{k}^i}^t\hat x(\tau)d\tau\\
	 & = \beta L_i^T\int_{t_{k}^i}^t 1(t-\tau)z(\tau)d\tau,\\
	\end{aligned}
		$$
	where $L_i^T$ denotes the $i$-th row of the Laplacian matrix and $1(t)=1$ when $t\geq0$ and $1(t)=0$ when $t<0$.
    Therefore, $e_i(t)=\phi_i(z)$ where $\phi_i(\cdot):\mathcal{L}_2\mapsto \mathcal{L}_2$ is a linear  operator in the time domain. 
	In light of  Lemma \ref{l2}, we have $e_i(s)=\Delta_i(z(s))$, where $\Delta_i(\cdot):\mathcal{L}_2\mapsto \mathcal{L}_2$ denotes a linear operator in the frequency domain and thus we have $e(s)=\Delta(z(s))$.

Note that $$\begin{aligned}
\|e\|_2^2&= \sum_{i=1}^N \|e_i\|^2_2
 = \sum_{i=1}^N\int_0^{\infty}e^*_i(t)e_i(t)dt.\\
\end{aligned}$$
According to  the triggering condition \dref{triggerfunction1}, for any time instant $t$, we have 
$$e_i^2\leq \alpha \sum_{j=1}^Na_{ij}(\hat x_i-\hat x_j)^2+\mu e^{-\nu t}.
$$
Therefore, $$\begin{aligned}
	 	\|e\|_{2}^2& \leq\sum_{i=1}^N\int_{0}^{\infty}\left(\alpha\sum_{j=1}^Na_{ij}(\hat x_i-\hat x_j)^2+\mu e^{-\nu t} \right)dt \\
	 	& =\int_{0}^{\infty}\left(\alpha\sum_{i=1}^N\sum_{j=1}^Na_{ij}(\hat x_i-\hat x_j)^2+N\mu e^{-\nu t} \right)dt \\
	 	& = \int_{0}^{\infty} \left(2 \alpha z^T L z +N\mu e^{-\nu t}\right) dt \\
	 	& \leq
	 	  2\alpha\lambda_N\|z\|_2^2+\frac{N\mu}{\nu}\\
	 	& \leq \left(\sqrt{2\alpha\lambda_N}\|z\|_2+\sqrt{\frac{N\mu}{\nu}}\right)^2.
	 \end{aligned}$$
 This completes the proof.
\end{proof}
\begin{remark}
The importance of this theorem lies in the following aspects. Firstly, it illustrates that the sampling errors can be seen as the images of linear finite-gain operators in the frequency domain acting on the consensus errors of the sampled states. This paves the way  to examine the robustness of the event-triggered algorithm to the frequency-domain uncertainties. These operators, generally speaking, are not rational transfer matrices in $\mathcal{RH}_{\infty}$  in the classic robust control. Nevertheless,  these operators are finite-gain $\mathcal L_2$ stable. 
Secondly, it uncovers  a quantitative relationship among the sampling parameter $\alpha$, the largest eigenvalue  of the Laplacian matrix, and the  $\mathcal L_2$ gain of the operators which quantifies  the effect of  aperiodic event-triggering. 
  More specifically,  the larger   $\lambda_N$  and  $\alpha$ are,  the larger the $\mathcal L_2 $ gain of the operators  will be.
\end{remark}
\begin{figure}
\centerline{\includegraphics[width=0.9\columnwidth,height=0.8\columnwidth]{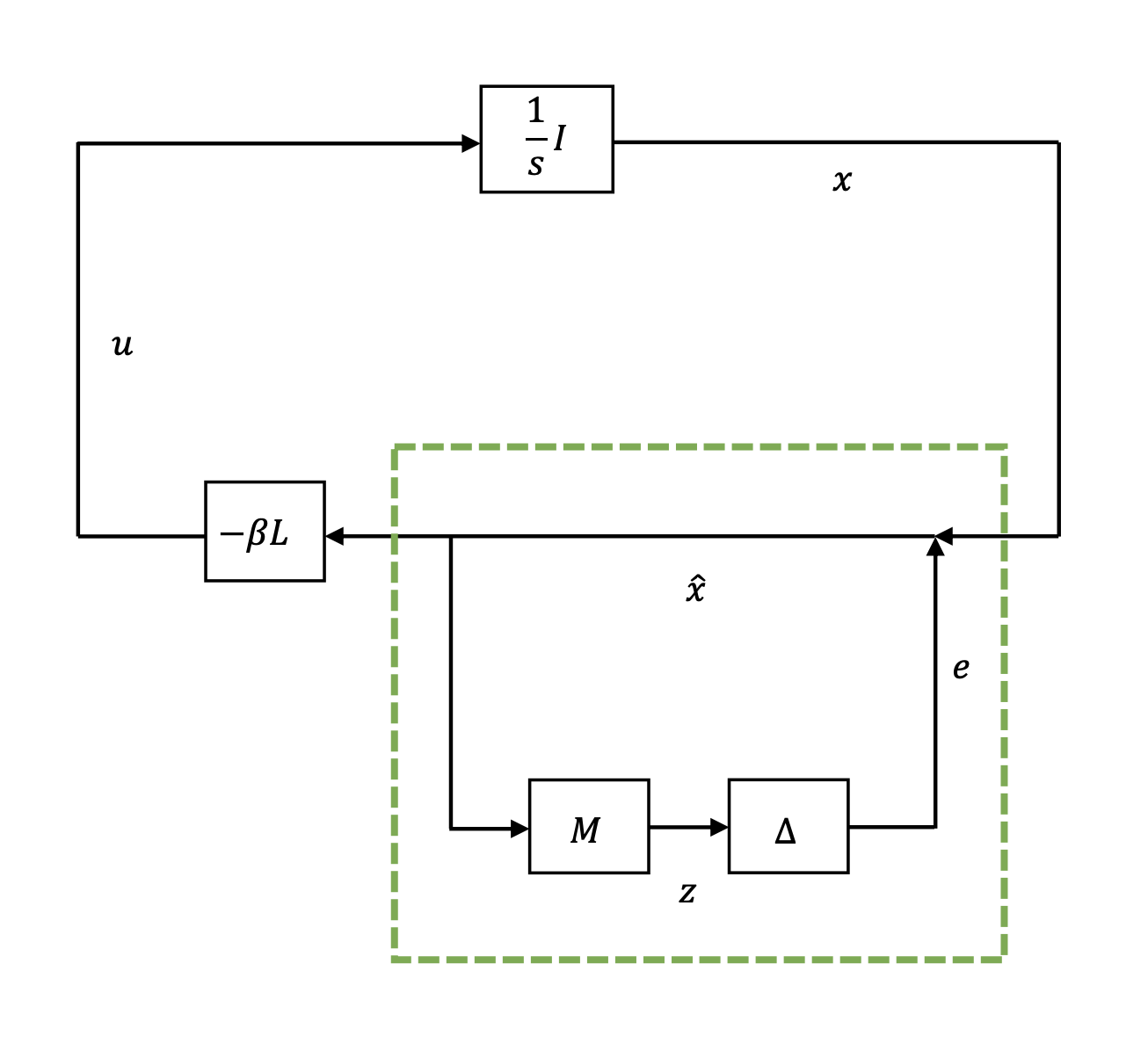}}
\caption{An operator theoretic reformulation of the distributed event-triggered consensus algorithm.}
\label{feedback}
\end{figure}

\begin{remark}
	The event-triggered consensus algorithm has an equivalent  block structure  shown in Fig. \ref{feedback}. Note that the effect of the event-triggered mechanism is actually equivalent to introducing the virtual additional feedback loop in the dotted green block. The  transfer function from $x$ to $\hat x$ is $R=(I-\Delta M)^{-1}$, i.e., $\hat x=(I-\Delta M)^{-1} x$.    When there is no event triggering, the operator $\Delta$ is equal to zero and thus  $R=I$. The event-triggered consensus algorithm \dref{eq3} then reduces to the classical one with continuous communication as in \cite{Olfati2004consensus}. 
\end{remark}

In the next theorem, a frequency-domain robust control  framework will be utilized 
to find the condition under which the network system reaches consensus by the event-triggered protocol \dref{eq2}. 
 
\begin{theorem}\label{th2}
	The network system \dref{eq1} reaches consensus under event-triggered control law \dref{eq2}, if the parameter $\alpha $ in the triggering function $f_i$ are selected to satisfy $2\alpha \lambda_N<1$. Moreover, the closed-loop system does not exhibit the  Zeno behavior as long as $\mu>0 $ and $\nu>0$. 
\end{theorem}
\begin{proof}
	Letting $U$ be the unitary matrix such that $U^TLU=\Lambda=\text{diag}\{0, \lambda_2, \cdots,\lambda_N\}=\text{diag}\{0,\bar\Lambda\}$, and denoting $\tilde x=U^Tx$, $\tilde z= U^Tz$, $\tilde e=U^Te$,   it is easy to find that  consensus is reached if and only if $\tilde x_i$ is asymptotically stable $\forall i=2,\cdots,N$. 
	Setting $\bar x=\tilde x_{2:N}$, $\bar z=\tilde z_{2:N}$, $\bar e=\tilde e_{2:N}$, where $(\cdot)_{2:N}$ denotes the subvector that takes the second to the $N$-th elements of the original vector. {   We can then derive}  from \dref{eq4} that
	\begin{equation}\label{eq6}
		\begin{aligned}
			& \dot {\bar x}=-\beta \bar\Lambda\bar x-\beta \bar\Lambda\bar e,\\
			&  \bar z= \bar x+\bar e.\\
		\end{aligned}
	\end{equation}
	Note also that $U$ can be written as $\begin{bmatrix}
		\frac{1_N}{\sqrt{N}} & Y 
	\end{bmatrix}$ with $Y^TY=I_{N-1}$ and $YY^T=M$. 
	It is not difficult to see that 
	$$\begin{aligned}
\tilde e&=U^T\Delta(UU^Tz)\\
&= \begin{bmatrix}
	\frac{1^T}{\sqrt{N}}\\
	Y^T
\end{bmatrix}\Delta(\frac{1}{\sqrt{N}}\tilde z_1+Y\bar z)\\
& = \begin{bmatrix}
	\frac{1^T}{\sqrt{N}}\\
	Y^T
\end{bmatrix}\Delta(Y\bar z)	,
\end{aligned}
$$  where we use $\tilde z_1\equiv 0$ to get the last equality.	
Therefore, we have \begin{equation}\label{eq7}
	\bar e=Y^T\Delta(Y\bar z)=\bar \Delta(\bar z),
\end{equation}
where $\bar\Delta$ is a linear  operator, and it is easy to verify that $$\begin{aligned}
	\|\bar e\|_2^2&=\int_0^{\infty}\Delta^*(Y\bar z)M\Delta(Y\bar z)dt\\
	&\leq\|\Delta(Y\bar z)\|_2^2\\
	& \leq\left( \sqrt{2\alpha\Lambda_N}\|Y\bar z\|_2+\sqrt{\frac{N\mu}{\nu}}\right)^2\\
	& =\left( \sqrt{2\alpha\lambda_N}\|\bar z\|_2+\sqrt{\frac{N\mu}{\nu}}\right)^2.
\end{aligned}$$
Therefore,
 $$\|\bar\Delta(\bar z)\|_2\leq\sqrt{2\alpha\lambda_N}\|\bar z\|_2+\sqrt{\frac{N\mu}{\nu}},$$	
which is equivalent to saying that $\|\bar\Delta\|_{\infty}\leq \sqrt{2\alpha\lambda_N}$.	
In light of Lemma \ref{smallgain},  the system is internally stable if $\|T_{\bar e\bar z}\|_{\infty}\|\bar\Delta\|_{\infty}<1$, where $T_{\bar e\bar z}$ denotes the transfer matrix from $\bar e$ to $\bar z$, calculated by
	$$T_{\bar e\bar z}=-(sI+\beta\bar\Lambda)^{-1}\beta\bar\Lambda+I=(sI+\beta\bar\Lambda)^{-1}sI.$$ Note that
	 $$\|T_{\bar e\bar z}\|_{\infty}=\max_{i=2,\cdots,N}\left\|\frac{s}{s+\beta\lambda_i}\right\|_{\infty}=1.$$
	Therefore, the system \dref{eq6} is internally stable, if $\|\bar \Delta\|_{\infty}<1$, which is satisfied if $2\alpha \lambda_N<1$. 
	
Next, we exclude the Zeno behavior. Notice that  during each time interval between any two consecutive triggering instants, i.e., $[t_{k}^i,t_{k+1}^i)$, $\dot e_i=-\dot x_i=\beta\sum_{j=1}^Na_{ij}(\hat x_i-\hat x_j)$ is bounded by a positive  real number, say $H$. Suppose that there exists Zeno behavior. Then there exists an agent $i$, such that $\lim_{k\rightarrow\infty}t_{k}^i=T<\infty$. Thus, for a small positive number $\delta<\frac{\sqrt{\mu }e^{-\frac{\nu T}{2}}}{H}$, there exists  a positive integer $K$ such that for $\forall k\geq K$, $t_{k}^i\in (T-\delta,T]$.  Notice that at the triggering instant $t_K^i$, $e_i=0$. And the next triggering time is the first time when $e_i^2$ reaches $\alpha\sum_{j=1}^Na_{ij}(\hat x_i-\hat x_j)^2+\mu e^{-\nu t}$. Then there must exist some time instant $T_0 $ when $e_i^2(T_0)=\mu e^{-\nu T_0}$. Since $$e_i^2(T_0)=(x_i(t_K^i)-x_i(T_0))^2=\left(\int_{t_K^i}^{T_0}\beta\sum_{j=1}^N(\hat x_i-\hat x_j)dt\right)^2,$$  
we have $$e_i^2(T_0)\leq H^2(T_0-t_K^i)^2.
$$ 
On the other hand, $\mu e^{-\nu T_0}\geq\mu e^{-\nu T}.  $  Thus we have $$H^2(T_0-t_K^i)^2\geq\mu e^{-\nu T},$$ which implies
$$
t_{K+1}^i-t_K^i\geq T_0-t_K^i\geq\frac{\sqrt{\mu }e^{-\frac{\nu T}{2}}}{H}.
$$
Since $\delta<\frac{\sqrt{\mu }e^{-\frac{\nu T}{2}}}{H}$, it follows that the $K+1$-th triggering instant $t_{K+1}>t_K^i+\delta>T$, which leads to a contradiction.
\end{proof}
\begin{remark}
	This theorem unveils some essential requirements to achieve consensus under event-triggered protocol \dref{eq2} that the gain (norm) of the operator should not be too large. The operator gain characterizes the extent of the sampling error introduced by the aperiodic event triggering. The larger the operator gain is, the larger the $\mathcal L_2$ norm of the sampling errors will be. 
From the quantitative relationship shown in this theorem, if  $\lambda_N$ is larger, then it is more reliable to trigger more frequently in the sense that $\alpha$  should be smaller.
\end{remark}
\begin{remark}
	 Different from most of the previous works, e.g.,\cite{cortes2016distributed}, \cite{dimarogonas2012distributed},\cite{Garcia2013decentralized}, \cite{yi2017distributed}, \cite{berneburg2019ditributed}, where time-domain Lyapunov stability analysis is used,  in this section an operator is constructed to characterize the relationship between the sampling error $e$ and the variable $z$. A robust control method based on the small gain theorem is utilized to get the consensus condition. Interestingly, since $\lambda_N\leq2d_i$, the consensus condition in Theorem \ref{th2} is less conservative than that of \cite{cortes2016distributed}. More importantly, this method provides a feasible way to handle the robustness of event-triggered control when the network systems are subject to frequency-domain uncertainties, as will be shown in the next section.            
	\end{remark}
	
\section{Robust Consensus Control via Event-Triggered Protocols}\label{s4}
In the last section, we analyze the event-triggered consensus problem of the integrator network without uncertainties in a frequency domain approach. One of the major merits of this approach is that it can handle various kinds of frequency-domain uncertainties in a unified framework. In this section we take additive dynamic uncertainties and network topology uncertainties as two  illustrating examples.             
\subsection{Additive Dynamic Uncertainties}
In this subsection, we consider the event-triggered consensus problem for the integrator network subject to additive dynamic uncertainties. The robust synchronization of linear multi-agent systems under such additive dynamic uncertainties with continuous communications was previously considered in  \cite{trentelman2013robust}.
Instead of measuring and exchanging  the state information directly at the triggering instants, each agent can only fetch an output variable $y_i$ that consists of the state variable $x_i$ and the disturbance signal $d_i$ caused by the dynamic perturbation $\Delta^a_i$. The agent dynamics are described by { 
\begin{equation}\label{eq8}
	\begin{aligned}
	&\dot x_i= u_i,\\
	& y_i= x_i+d_i,\\
	& d_i(s)=\Delta^a_i (u_i(s)),\quad i=1,\cdots,N\\
	\end{aligned}	
\end{equation}
where $\Delta_i^a$ is a linear} { finite-gain $\mathcal L_2$ stable operator with operator gain $\|\Delta^a_i\|_{\infty}\leq \eta$. Note that this definition include transfer matrices in $\mathcal{RH}_{\infty}$ as special cases.} 

We consider the following  distributed output feedback protocol:
\begin{equation}\label{eq9}
	u_i=-\beta \sum_{j=1}^Na_{ij}(\hat y_i-\hat y_j), \quad i=1,\cdots,N,
\end{equation} 
where $\hat y_i$ denotes the estimate of the output $y_i$, which is updated  to the real value $y_i(t_{k}^i)$ and broadcasted to all its neighbors at the $k$-th triggering instant of the agent $i$, i.e., $t_{k}^i$ and keeps constant (ZOH) during two triggering instants.
Define $\epsilon_i=\hat y_i-y_i$. Moreover, the triggering function of agent $i$ is set to be 
\begin{equation}\label{triggerfunction2}
	f_i=\epsilon_i^2-\alpha\sum_{j=1}^Na_{ij}(\hat y_i-\hat y_j)^2-\mu e^{-\nu t}.
\end{equation}

According to \dref{eq8} and \dref{eq9}, we have
\begin{equation}\label{eq10}
	\dot x_i= -\beta\sum_{j=1}^Na_{ij}\left[(x_i-x_j)+(d_i-d_j)+(\epsilon_i-\epsilon_j)\right].
\end{equation}
Denote $w_i=\frac{1}{N}\sum_{j=1}^N(\hat y_i-\hat y_j)$. We can obtain the closed-loop network dynamics as follows:
\begin{equation}\label{eq12}
	\begin{aligned}
		& \dot x = -\beta L x-\beta L d-\beta L \epsilon,\\
		& u = -\beta L x-\beta L d-\beta L \epsilon,\\
		& w = M x+M d+M \epsilon .\\
	\end{aligned}
\end{equation}

\begin{theorem}\label{th4}
	For the  network \dref{eq8} under the control law  \dref{eq9} with the triggering  function \dref{triggerfunction2}, it follows that $\epsilon(s)=\Delta^b(w(s))$, where $\Delta^b(\cdot):\mathcal{L}_2\mapsto\mathcal{L}_2$ is a linear operator with the operator gain $\|\Delta^b\|_{\infty}\leq\sqrt{2\alpha \lambda_N}$.
\end{theorem}

\begin{proof}
	From the definition of $\epsilon_i$ we know that  $\epsilon_i(t)=\hat y_i-y_i=y(t_{k}^i)-y_i(t)=x_i(t_{k}^i)+d_i(t_{k}^i)-x_i(t)-d_i(t)$.
	According to \dref{eq8} and \dref{eq9}, when $t\in [t_{k}^i,t_{k+1}^i)$, 
	$$\begin{aligned}
		x_i(t)-x_i(t_{k_i}^i)&=-\beta \int_{t_{k_i}^i}^t\sum_{j=1}^Na_{ij}(\hat y_i(\tau)-\hat y_j(\tau))d\tau\\&=-\beta\int_{t_{k}^i}^tL_i^T\hat y(\tau)d\tau\\&=-\beta L_i^T\int_{t_{k}^i}^tw(\tau)d\tau.
	\end{aligned}$$
	Thus, we have $x_i(t)-x_i(t_{k_i}^i)=\psi_i(w(t))$, 
	 where $\psi_i(\cdot):\mathcal{L}_2\mapsto\mathcal{L}_2$ is a linear  operator.
	{ Since $d_i(s)=\Delta_i^a(u_i(s))$, we let $d_i(t)=\tilde{\Delta}_i^a(u_i(t))$ and   $d(t_k^i)=\bar\Delta_i^a(u_i(t))$, and we have that $\tilde{\Delta}_i^a(\cdot)$ and $\bar\Delta_i^a(\cdot)$ are linear operators in the time domain. }
	It is then not difficult to find that $d_i(t_{k_i}^i)-d_i(t)$ is a linear operator of $u_i(t)$. 
	
	Denote $d_i(t_{k}^i)-d_i(t)=\phi_i(u_i(t))$ and note that $u_i(t)=-\beta\sum_{j=1}^Na_{ij}(\hat y_i(t)-\hat y_j(t))=-\beta L_i^T\hat y(t)=-\beta L_i^Tw(t)$.
	Thus $$d_i(t_{k}^i)-d_i(t)=\phi_i(-\beta L_i^Tw(t))=\phi_i'(w(t)),$$ 
	and it is evident that $\phi_i'(\cdot):\mathcal{L}_2\mapsto\mathcal{L}_2$ is a linear  operator. 
	Therefore, $\epsilon_i=-\psi_i(w(t))+\phi'_i(w(t))=\Psi_i(w(t)),$ where $\Psi_i'(\cdot):\mathcal{L}_2\mapsto\mathcal{L}_2$ is a linear  operator. According to Lemma \ref{l2}, $\epsilon_i(s)=\Delta^b_i(w(s))$, where $\Delta_i^b(\cdot):\mathcal{L}_2\mapsto\mathcal{L}_2$ is a linear operator in the frequency domain. Moreover, $\epsilon=\Delta^b(w(s))$, where $\Delta :\mathcal{L}_2\mapsto\mathcal{L}_2$ is a linear operator. The operator gain $\|\Delta^b\|_{\infty}$ can be similarly determined as in  Theorem \ref{th1} and  is omitted here for brevity. 
\end{proof}

\begin{figure}[!t]
\centerline{\includegraphics[width=0.9\columnwidth,height=0.8\columnwidth]{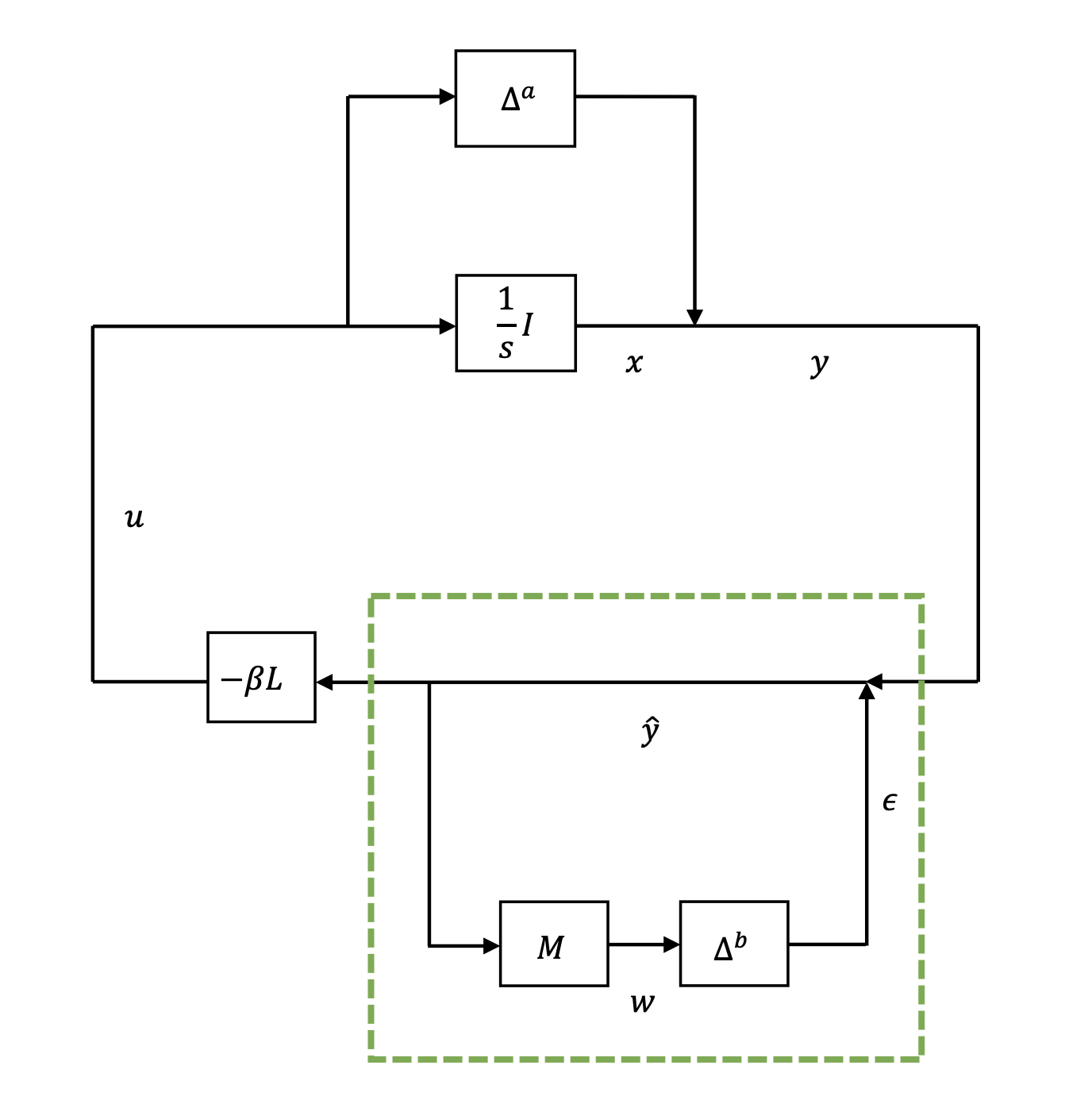}}
\caption{Event-triggered network system with additive dynamic uncertainties.}
\label{robust}
\end{figure}
\begin{remark}
 	It is worth noting that  two blocks of operators $\Delta^a $ and $\Delta^b$ appearing  in Fig. \ref{robust} are essentially different. The first block of operator represents the uncertainties of the agent dynamics in $\mathcal{RH}_{\infty}$, possibly caused by model uncertainties, unmodeled dynamics, or nonlinear behavior of the agent itself. The second block of operator $\Delta^b$ is caused by the event-triggered sampling mechanism and in general does not belong to $\mathcal{RH}_{\infty}$.  Quite interestingly, these two blocks of operators generated by totally different mechanisms  can be unified in an operator-theoretic framework, since the two blocks of uncertainties are both linear finite-gain $\mathcal L_2$ stable operators in the frequency domain.   
\end{remark}

As will shown in  the the next theorem, the two blocks of operators in the Fig. \ref{robust} cannot be too large in order to guarantee the robust consensus of the network system. 

\begin{theorem}\label{th44}
	Let $\gamma=\max \{\eta,\sqrt{2\alpha \lambda_N}\}$. The uncertain network \dref{eq8} reaches robust  consensus under event-triggered  control law \dref{eq9}  with the  triggering function \dref{triggerfunction2}, if  the positive constants $\beta$, $\alpha$, $\mu$, and $\nu$  satisfy that $(\beta\lambda_N+1)\gamma<1$. Moreover, the closed-loop system does not exhibit the Zeno behavior. 
\end{theorem}
\begin{proof}
	Notice that the interconnecting system can be rewritten into the following form:
	\begin{equation}\label{eq13}
		\begin{aligned}
			\dot x&= -\beta L x+\begin{bmatrix}
				-\beta L & -\beta L
			\end{bmatrix}\begin{bmatrix}
				d\\ \epsilon
			\end{bmatrix},\\
			\begin{bmatrix}
				u\\ w
			\end{bmatrix}
			&=\begin{bmatrix}
				-\beta L\\ M
			\end{bmatrix}x
			+\begin{bmatrix}
				-\beta L & -\beta L\\
				M & M\\
			\end{bmatrix}\begin{bmatrix}
				d\\ \epsilon
			\end{bmatrix},\\
			d(s)&=\Delta^a (u(s)),\\
			\epsilon(s)& = \Delta^b(w(s)).
		 		\end{aligned}
	\end{equation}
	Here $\Delta^a(\cdot)=\mathrm{diag}\{\Delta_1^a,\cdots,\Delta_N^a\}$ is also a linear  operator with operator gain $\|\Delta^a\|_{\infty}\leq \eta$.
	Similarly to Theorem \ref{th2}, denoting $\tilde \epsilon=U^T\epsilon$, $\tilde w=U^Tw$, $\bar\epsilon=\tilde\epsilon_{2:N}$, $\bar w=\tilde w_{2:N}$, we can get that the system \dref{eq13} reaches robust consensus if  and only if the following system interconnection:
	\begin{equation}\label{eq14}
	\begin{aligned}
		\dot{\bar x}&=-\beta \bar\Lambda\bar x +\begin{bmatrix}
				-\beta \bar\Lambda & -\beta \bar\Lambda
			\end{bmatrix}\begin{bmatrix}
				\bar d\\ \bar\epsilon
			\end{bmatrix},\\
			\begin{bmatrix}
				\bar u\\ \bar w
			\end{bmatrix}
			&=\begin{bmatrix}
				-\beta \bar\Lambda\\ I
			\end{bmatrix}\bar x
			+\begin{bmatrix}
				-\beta \bar\Lambda & -\beta \bar\Lambda\\
				I & I\\
			\end{bmatrix}\begin{bmatrix}
				\bar d\\ \bar\epsilon
			\end{bmatrix},\\
			\bar d(s)&=Y^T\Delta^a (Y\bar u(s))=\bar\Delta^a(\bar u(s)),\\
			\bar \epsilon(s)& = Y^T\Delta^b(Y\bar w(s))=\bar\Delta^b(\bar w(s)) 
	\end{aligned}
	\end{equation}
	is internally stable, where $\bar\Delta^a(\cdot):\mathcal{L}_2\mapsto\mathcal{L}_2$ and $\bar\Delta^b(\cdot):\mathcal{L}_2\mapsto\mathcal{L}_2$ are  linear operators with operator gains $\|\bar\Delta^a\|_{\infty}\leq \eta$ and $\|\bar\Delta^b\|_{\infty}\leq\sqrt{2\alpha\lambda_N}$, respectively. 
	Note that we can further derive that \begin{equation}
		\begin{aligned}
			\begin{bmatrix}
				\bar d(s)\\
				\bar \epsilon(s)
			\end{bmatrix}=\mathbf{\Delta}\begin{bmatrix}
				\bar u(s)\\ \bar w(s)
			\end{bmatrix},
		\end{aligned}
	\end{equation}
	where $\mathbf{\Delta}$ is a linear block diagonal operator defined as $$\mathbf\Delta=\begin{bmatrix}
		\bar\Delta^a(\cdot) & 0(\cdot)\\
		0(\cdot) & \bar\Delta^b(\cdot)
	\end{bmatrix}$$ and $\|\mathbf \Delta\|_{\infty}\leq\gamma$.
	In light of Lemma \ref{l3}, the system \dref{eq14} reaches  internal stability if $\mu_{\mathbf\Delta}(G(j\omega))<\frac{1}{\gamma}$, where $G$ is the transfer matrix from 
	$\begin{bmatrix}
		\bar d^T  & \bar \epsilon^T
	\end{bmatrix}^T$ to $\begin{bmatrix}
		\bar u^T  & \bar w^T
	\end{bmatrix}^T$.
	From \dref{eq14}, it is easy to derive that $$\begin{aligned}
		G(s)&=\begin{bmatrix}
		-\beta s \bar \Lambda (sI+\beta\bar\Lambda)^{-1} & -\beta s \bar \Lambda (sI+\beta\bar\Lambda)^{-1}\\
		s(sI+\beta\bar\Lambda)^{-1} & s(sI+\beta\bar\Lambda)^{-1}
	\end{bmatrix}\\
	&\triangleq \begin{bmatrix}
		G_{11}(s) & G_{12}(s)\\
		G_{21}(s) & G_{22}(s)
	\end{bmatrix},
	\end{aligned} $$
	 where 
    $$\begin{aligned}
    	\|G_{11}(j\omega)\|^2&=\|G_{12}(j\omega)\|^2\\
	&=\|j\omega\beta\bar\Lambda(j\omega I+\beta\bar\Lambda)^{-1}\|^2\\
    	& =  \max_{i=2,\cdots,N}\frac{\beta^2\lambda_i^2\omega^2}{\omega^2+\beta^2\lambda_i^2}= \frac{\beta^2\lambda_N^2\omega^2}{\omega^2+\beta^2\lambda_N^2},
    \end{aligned}$$
    and
    $$\begin{aligned}
    	\|G_{22}(j\omega)\|^2&=\|G_{21}(j\omega)\|^2\\&=\|j\omega (j\omega I+\beta\bar\Lambda)^{-1}\|^2= \frac{\omega^2}{\omega^2+\beta^2\lambda_2^2}.
    \end{aligned}$$
    According to Lemma \ref{l4}, the system interconnection \dref{eq14} is internally stable, if  $\frac{\beta\lambda_N|\omega|}{\sqrt{\omega^2+\beta^2\lambda_N^2}}+\frac{|\omega|}{\sqrt{\omega^2+\beta^2\lambda_2^2}}<\frac{1}{\gamma}$, which holds if $\beta\lambda_N+1<\frac{1}{\gamma}$. Next, it remains to rule out the Zeno behavior. This procedure is quite similar to the counterpart  of Theorem \ref{th2}, and is omitted here for conciseness.  
    \end{proof}
\begin{remark}
Though the time-domain sampling mechanism  and the frequency-domain uncertainties are incompatible seemingly, they can indeed be tackled in a unified way under our operator-theoretic framework. This theorem quantitatively characterizes the relationship among Laplacian matrix, the output feedback gain  and the robustness of the event-triggered consensus algorithm against event triggering and frequency-domain uncertainties.
From this theorem, it is not difficult to find that the smaller  $\lambda_N$ and the parameter $\beta$ are, the larger $\gamma$ can be, meaning that  the more robust the network is against frequency-domain uncertainties  and sparser triggering. 
\end{remark}
\begin{remark}
   To the our best knowledge, in  most of the literatures in this realm, e.g.,\cite{cortes2016distributed},  \cite{dimarogonas2012distributed}, \cite{Garcia2013decentralized}, \cite{yi2017distributed}, \cite{berneburg2019ditributed}, frequency-domain uncertainties have not been addressed yet. The operator-theoretic approach proposed here links for the first time the time-domain sampling mechanism to the frequency-domain uncertainties, and lays the foundation of the robust analysis and synthesis of event-triggered control.              
\end{remark} 

\subsection{Network Topology Uncertainties }

In the previous parts of this paper, it is assumed that the agents interact via a fixed and known graph $\mathcal G$. In practice, however, the topology graph may be subject to various kinds of perturbations, which will render the network graph uncertain and  time-varying \cite{lirobust2017}. 
In this subsection, we consider a network graph with uncertain communication strengths. Specifically, each  element $a_{ij}$ in the adjacency matrix  cannot be known exactly   but is rather perturbed to a bounded region that can be expressed in the form of  $a_{ij}(1+\Delta^c_{ij}(\cdot))$, where $\Delta_{ij}^c(\cdot):\mathcal L_2\mapsto \mathcal L_2$ is a linear finite-gain operator in the time domain with operator norm $\|\Delta_{ij}^c\|_{\infty}\leq\delta<1$.  
In this case, distributed control law \dref{eq2} then becomes 
\begin{equation}\label{eq16}
\begin{aligned}
	u_i &=  -\beta\sum_{j=1}^Na_{ij}(\hat x_i-\hat x_j)-\beta \sum_{j=1}^Na_{ij}\Delta^c_{ij}(\hat x_i-\hat x_j)\\
	& = -\beta\sum_{j=1}^Na_{ij}(\hat x_i-\hat x_j)-\beta\sum_{j=1}^N(a_{ij}^{\frac{1}{2}}\Delta^c_{ij}(a_{ij}^{\frac{1}{2}}\hat x_i)\\
	& \quad -a_{ij}^{\frac{1}{2}}\Delta^c_{ij}(a_{ij}^{\frac{1}{2}}\hat x_j))\\
	& = -\beta L_i^T\hat x-\beta D_iW^{\frac{1}{2}}\Delta^c(W^{\frac{1}{2}}D^T\hat x),
\end{aligned}
\end{equation} 
where $D_i$ is the $i$-th row of the incidence matrix of the graph $\mathcal G$ and $\Delta^c(\cdot)$ is the block diagonal linear finite-gain operator with $\Delta^c_{ij}(\cdot)$ as its diagonal elements. It is easy to see that $\|\Delta^c\|_{\infty}\leq\delta$. Letting  $v=\Delta^c(W^{\frac{1}{2}}D^T\hat x)$ and $w=W^{\frac{1}{2}}D^T\hat x$, we have $v=\Delta^c(w)$. According to Lemma \ref{l2}, $v(s)=\bar\Delta^c(w(s))$. Define  the sampling error $e_i=\hat x_i-x_i$.
The triggering function of the agent $i$ is \begin{equation}\label{eq18}
	f_i=e_i^2-\alpha\sum_{j=1}^Na_{ij}(\hat x_i-\hat x_j)^2-\mu e^{-\nu t}.
\end{equation} 
 Recalling \dref{eq1}, then the system interconnection can be rewritten in the following compact form:
\begin{equation}\label{eq17}
\begin{aligned}
	&\dot x = -\beta Lx-\beta L e-\beta DW^{\frac{1}{2}}v,\\
	&w =W^{\frac{1}{2}}D^Tx+W^{\frac{1}{2}}D^Te,\\
	& z = M x+Me,\\
	& v(s)= \bar\Delta^c(w(s)).
\end{aligned}
\end{equation} 
Similarly, we have   $$\begin{aligned}
	e_i(t)
	=&\beta \int_{t_{k}^i}^tL_i^T\hat x(\tau)d\tau+\beta\int_{t_{k}^i}^tD_iW^{\frac{1}{2}}\Delta^c(W^{\frac{1}{2}}D^T\hat x(\tau))d\tau\\
	=& \beta \int_{t_{k}^i}^tD_iWz(\tau)d\tau+\beta\int_{t_{k}^i}^tD_iW^{\frac{1}{2}}\Delta^c(W^{\frac{1}{2}}z(\tau))d\tau\\
	=& \Delta^d(z(t)),
\end{aligned}$$
where $\Delta^d(\cdot):\mathcal L_2\mapsto\mathcal L_2$ is a linear operator in the time domain, following the proof of Theorem \ref{th1}.  In virtue of Lemma \ref{l2}, $e_i(s)=\bar\Delta^d(z(s))$, where $\bar\Delta^d(\cdot):\mathcal L_2\mapsto\mathcal L_2$ is a linear operator in the frequency domain.
Then, in light of Theorem \ref{th1}, it is easy to see that $\|\bar\Delta\|^d_{\infty}\leq \sqrt{2\alpha \lambda_N}$.
Moreover, $$\begin{aligned}
	\|\bar\Delta^c\|^2_{\infty}&=\sup_{\|w\|_2=1}\frac{1}{2\pi}\int_{-\infty}^{\infty}v^*(j\omega)v(j\omega)d\omega\\
	& = \sup_{\|w\|_2=1}\int_0^{\infty}v^*(t)v(t)dt\\
	& = \|\Delta^c\|^2_{\infty}\leq\delta^2.
\end{aligned}$$
We have $\|\bar\Delta^c\|_{\infty}\leq\delta.$

\begin{theorem}\label{th5}
	 The  network   \dref{eq1} reaches consensus  under  the event-triggered control law \dref{eq16}  with the triggering function   \dref{eq18}, if $\gamma<\sqrt{\frac{\lambda_2^2}{\lambda_2^2+\lambda_N^2}}$ where $\gamma=\max \{\delta,\sqrt{2\alpha \lambda_N} \}$. Moreover, the closed-loop system does not exhibit the Zeno behavior.
\end{theorem}
\begin{proof}
	Letting $\xi=Mx$, the system interconnection can be written in the following form:
	\begin{equation}\label{eq19}
	\begin{aligned}
		& \dot \xi=-\beta L\xi +\begin{bmatrix}
			-\beta DW^{\frac{1}{2}} & -\beta L
		\end{bmatrix} 
		\begin{bmatrix}
			v\\ e
		\end{bmatrix},\\
		& \begin{bmatrix}
			w\\ z
		\end{bmatrix}=
		\begin{bmatrix}
			W^{\frac{1}{2}}D^T\\ M 
		\end{bmatrix}\xi+
		\begin{bmatrix}
			0 & W^{\frac{1}{2}}D^T\\
			0& M
		\end{bmatrix}
		\begin{bmatrix}
			v\\ e
		\end{bmatrix},\\
		& \begin{bmatrix}
			v(s)\\ e(s)
		\end{bmatrix}=
		\begin{bmatrix}
			\bar\Delta^c &\\
			& \bar\Delta^d
		\end{bmatrix}
		\begin{bmatrix}
			w(s)\\ z(s)
		\end{bmatrix}.
	\end{aligned}
	\end{equation}
It is also worth noting that the state $x$ of the system reaches consensus if and only if $\xi$ is asymptotically stable, i.e., the system interconnection \dref{eq19} is internally stable.
According to   Lemma \ref{l3}, it is enough to let $\mu_{\mathbf\Delta}(G(j\omega))<\frac{1}{\gamma}$, where $G$ is the transfer matrix from $\begin{bmatrix}
	v^T & e^T
\end{bmatrix}^T $ to  $\begin{bmatrix}
	w^T & z^T
\end{bmatrix}^T $ and 
$$\begin{aligned}
	G  &\triangleq  \begin{bmatrix}
 	G_{11}(s) & G_{12}(s)\\
 	G_{21}(s) & G_{22}(s)
 \end{bmatrix}\\
	&=\begin{bmatrix}
	-\beta W^{\frac{1}{2}}D^T(sI+\beta L)^{-1} DW^{\frac{1}{2}}  & sW^{\frac{1}{2}}D^T(sI+\beta L)^{-1}\\
	-\beta M(sI+\beta L)^{-1}DW^{\frac{1}{2}} & sM(sI+\beta L)^{-1}
\end{bmatrix}.\\
\end{aligned}
$$
Note that
$$\begin{aligned}
	&\|G_{11}(j\omega)\|^2  \\
	&\qquad=\rho(\beta^2W^{\frac{1}{2}}D^T(-j\omega I+\beta L)^{-1}L(j\omega I+\beta L)^{-1} DW^{\frac{1}{2}})\\
		&\qquad= \rho(\beta^2 (-j\omega I+\beta \Lambda)^{-1}\Lambda(j\omega I+\beta \Lambda)^{-1}\Lambda)\\
	& \qquad= \frac{\beta^2\lambda_N^2}{\omega^2+\beta^2\lambda_N^2}\leq \frac{\beta^2\lambda_N^2}{\omega^2+\beta^2\lambda_2^2}.
\end{aligned}$$
Similarly, $$\|G_{12}(j\omega)\|^2\leq\frac{\omega^2\lambda_N}{\omega^2+\beta^2\lambda_2^2},$$ $$\|G_{21}(j\omega)\|^2=\max_{i=2,\cdots,N}\frac{\beta^2\lambda_i}{\omega^2+\beta^2\lambda_i^2}\leq\frac{\beta^2\lambda_N}{\omega^2+\beta^2\lambda_2^2},$$
and
$$\|G_{22}(j\omega)\|^2=\frac{\omega^2}{\omega^2+\beta^2\lambda_2^2}.$$
According to Lemma \ref{l4}, the system interconnection is internally stable, if  
\begin{equation}\label{add}
\frac{(\beta\lambda_N+|\omega|)^2}{\beta^2\lambda_2^2+\omega^2}<\frac{1}{\gamma^2},\quad\forall \omega\in \mathbf R.	
\end{equation}
It is not difficult to find that a sufficient condition for \dref{add} to hold is $\gamma<\sqrt{\frac{\lambda_2^2}{\lambda_2^2+\lambda_N^2}}$. The excluding of the Zeno behavior is quite similar to the counterpart of Theorem \ref{th2} and  is omitted here for brevity. 
\end{proof}
\begin{remark}
	{  Different from the additive dynamic uncertainty case considered in the previous subsection, the robustness of the event-triggered consensus algorithm against the network topology uncertainties is closely related to the eigenratio $\frac{\lambda_N}{\lambda_2}$ of the Laplacian matrix. When the eigenratio $\frac{\lambda_N}{\lambda_2}$ is smaller, the robust margin $\gamma$ is larger in the sense that the network can tolerant larger network uncertainties and larger sampling errors. }
\end{remark}
%
\begin{remark}
Note that the uncertainties of the network topology can be either time-domain or frequency-domain, which are equivalent according to Lemma \ref{l2}. Actually, it is not necessary for the uncertainties to be in $\mathcal{RH}_{\infty}$. As long as the uncertainties  are finite-gain $\mathcal L_2$-stable operators, e.g., bounded communication delays or nonlinearities, the robust event-triggered consensus problem can be solved in our operator-theoretic framework.
\end{remark}

\section{Extension to Dynamic Average Consensus }\label{s5}
In the previous sections, the control objective is to reach static average consensus, i.e., each state variable $x_i$ converges to the average of the initial states of all agents. In this section, we generalize the results to the event-based dynamic average consensus (DAC) problem. { Different from the static average consensus problem, the DAC problem  aims to make the state of  each agent converge to the average of the reference signals \cite{kia2019tutorial}.} Till now,  there have been a lot of DAC algorithms proposed in the literatures for first-order network systems, e.g.,  \cite{freeman2006stability}, \cite{kia2019tutorial}. An  event-triggered DAC algorithm is provided in  \cite{kia2015distributed}, where the agent model is assumed to be be nominal and there are no uncertainties.  Here, we consider the following  network subject to additive dynamic uncertainties: { 
\begin{equation}\label{eq20}
	\begin{aligned}
		&\dot x_i = u_i +\dot r_i,\\
		& y_i= x_i+d_i,\\
		& d_i(s) = \Delta_i^a (u_i(s)),\quad	i =1,\cdots, N,
	\end{aligned}
\end{equation}}
where $r_i$ denote the reference signals and $\Delta_i^a$ is defined as in \dref{eq8}.  { The objective of the robust DAC problem considered in this section is to ensure that for the uncertain agents in \dref{eq20}, $x_i \rightarrow \frac{1}{N}\sum_{i=1}^Nr_i, \forall i=1,\cdots,N$ as $t\rightarrow\infty$.} Throughout this section, we suppose that the following assumption holds.
\begin{assumption}
	The reference signals $r_i$ are bounded. 
\end{assumption}

{ To achieve DAC, the following event-triggered algorithm is utilized:}
\begin{equation}\label{eq21}
	\begin{aligned}
		&\dot w_i = -\theta(w_i-y_i)-\beta\sum_{j=1}^Na_{ij}( \hat w_i- \hat w_j),\\
		& u_i = -\beta \sum_{j=1}^Na_{ij}( \hat w_i- \hat w_j),
	\end{aligned}
\end{equation}
where $w_i$  is the augmented state variable. To save the communication cost, the event-triggered scheme  is used.  At the $k$-th triggering instant $t_{k}^i$ of the agent $i$, the $i$-th agent  updates its estimate of its own  states $\hat w_i$  to $w(t_{k}^i)$ and its  estimate of its own output $\hat y_i$  to $y(t_{k}^i)$ and broadcasts them to all its neighbors. During two triggering instants of the agent $i$, $\hat w_i$ is calculated in the following way:
$$
\hat w_i(t)= e^{-\theta(t-t_{k}^i)}w(t_{k}^i)+\int_{t_{k}^i}^te^{-\theta(t-\tau)}\theta \hat y_i(\tau)d\tau,
$$
and $$
\hat y_i(t) = y(t_{k}^i).
$$
The triggering function of the agent $i$ is given by 
\begin{equation}\label{eq22}
	f_i= \epsilon_i^2-\alpha \sum_{j=1}^Na_{ij}(\hat w_i-\hat w_j)^2-\mu e^{-\nu t},
\end{equation}
where $\epsilon_i=\hat w_i-w_i$ denotes the sampling error. Define $z_i=\frac{1}{N}\sum_{j=1}^N(\hat w_i-\hat w_j)$, $z=\begin{bmatrix}
	z_1& \cdots,z_N
\end{bmatrix}^T$ and $\epsilon =\begin{bmatrix}
	\epsilon_1& \cdots,\epsilon_N
\end{bmatrix}^T$. It then follows that $z=M\hat w$. 
Similarly, we have the following claim about the relationship between $\epsilon$ and $z$.
\begin{theorem}
{ 	For the  triggering function \dref{eq22}, it  follows that $\epsilon(s)=\Delta^e(z(s))$, where $\Delta^e(\cdot)$ is a nonlinear finite-gain $\mathcal L_2$ stable operator with the operator norm $\|\Delta^e\|_{\infty}\leq\sqrt{2\alpha \lambda_N}$. }
\end{theorem}
\begin{proof}
	It follows from  \dref{eq21} that for $t\in [t_{k}^i,t_{k+1}^i)$, $$\begin{aligned}
		w_i(t)=&e^{-\theta(t-t_{k}^i)}w(t_{k}^i)+\int_{t_{k}^i}^te^{-\theta(t-\tau)}\theta  y_i(\tau)d\tau\\
		& -\int_{t_{k}^i}^te^{-\theta(t-\tau)}\beta\sum_{j=1}^Na_{ij}(\hat w_i(\tau)-\hat w_j(\tau))d\tau.
	\end{aligned}$$
	Therefore, $$\begin{aligned}
		\epsilon_i =& \hat w_i-w_i\\=& \int_{t_{k}^i}^te^{-\theta(t-\tau)}\beta\sum_{j=1}^Na_{ij}(\hat w_i(\tau)-\hat w_j(\tau))d\tau\\
		& + \int_{t_{k}^i}^te^{-\theta(t-\tau)}\theta  (\hat y_i(\tau)-y_i(\tau))d\tau\\
		=& \int_{t_{k}^i}^te^{-\theta(t-\tau)}\theta  ((\hat y_i(\tau)-y_i(\tau))-(\hat r_i(\tau)-r_i(\tau)))d\tau\\
		&+\int_{t_{k}^i}^te^{-\theta(t-\tau)}\theta  (\hat r_i(\tau)-r_i(\tau))d\tau\\
		& + \int_{t_{k}^i}^te^{-\theta(t-\tau)}\beta\sum_{j=1}^Na_{ij}(\hat w_i(\tau)-\hat w_j(\tau))d\tau,
	\end{aligned}$$
	where $\hat r_i(t)=r_i(t_{k}^i)$, $\forall t\in [t_{k}^i,t_{k+1}^i)$.
	On the other hand,
	 $$\begin{aligned}
		\hat y_i-y_i&=x(t_{k}^i)+d(t_{k}^i)-(x(t)+d(t))\\
		& = d(t_{k}^i)-d(t)+\beta\int_{t_{k}^i}^t \sum_{j=1}^Na_{ij}(\hat w_i-\hat w_j )d\tau\\
		&\quad-r_i(t)+r_i(t_{k}^i).\\
	\end{aligned}$$
	{ Since $d_i(s)=\Delta_i^a(u_i(s))$, we have $d(t)=\tilde{\Delta}_i^a(u_i(t))$ and   $d(t_k^i)=\bar\Delta_i^a(u_i(t))$, where $\tilde{\Delta}_i^a(\cdot)$ and $\bar\Delta_i^a(\cdot)$ are linear operators in the time domain. }
	Notice that { 
	$$\begin{aligned}
		d_i(t_{k}^i)-d_i(t)=&\bar\Delta_i^a(u_i(t))-\tilde{\Delta}_i^a(u_i(t))\\
		=& \Theta_i(u_i(t)) \\
		=& \Theta_i(-\beta\sum_{j=1}^Na_{ij}(\hat w_i-\hat w_j))\\
		=& \Xi_i (z(t))
	\end{aligned}
	$$ } where $\Theta,\Xi(\cdot):\mathcal L_2\mapsto\mathcal L_2$ are linear operators. 
	Therefore, $$\begin{aligned}
		\hat y_i-y_i-(\hat r_i-r_i)&=\beta\int_{t_{k}^i}^t \sum_{j=1}^Na_{ij}(\hat w_i-\hat w_j )d\tau\\
		&\quad+d(t_{k_i}^i)-d(t)\\&=\beta L_i^T\int_{t_{k}^i}^tz(\tau)d\tau+\Xi_i(z(t))\\&=\Phi_i(z(t)),
	\end{aligned}$$  where $\Phi_i(\cdot):\mathcal L_2\mapsto \mathcal L_2$ is a linear operator. Note also that $$\begin{aligned}
		&\int_{t_{k}^i}^te^{-\theta(t-\tau)}\beta\sum_{j=1}^Na_{ij}(\hat w_i(\tau)-\hat w_j(\tau))d\tau\\&\qquad=\beta L_i^T\int_{t_{k}^i}^te^{-\theta(t-\tau)}z(\tau)d\tau=\Gamma_i(z(t)),
	\end{aligned}$$
	where $\Gamma_i(\cdot):\mathcal L_2\mapsto \mathcal L_2$ is a linear operator.
	We then have $$\begin{aligned}
		\epsilon_i(t)=&\int_{t_{k}^i}^te^{-\theta(t-\tau)}\theta  \Phi_i(z(\tau))d\tau+\Gamma_i(z(t))\\
		&+\int_{t_{k_i}^i}^te^{-\theta(t-\tau)}\theta  (\hat r_i(\tau)-r_i(\tau))d\tau\\
		=& \Psi_i(z(t))+\int_{t_{k}^i}^te^{-\theta(t-\tau)}\theta  (\hat r_i(\tau)-r_i(\tau))d\tau,
	\end{aligned}$$
	where $\Psi_i(\cdot):\mathcal L_2\mapsto \mathcal L_2$ is a linear operator. Since  $r_i$ is bounded $\forall t>0$, {  then we have $\epsilon_i(s)=\Delta^e_i(z(s))$ where $\Delta^e_i(\cdot):\mathcal L_2\mapsto \mathcal L_2$ is a nonlinear operator in the frequency domain.  It then remains to determine the operator gain $\|\Delta^e\|_{\infty}$. This procedure is quite similar to the counterpart of Theorem \ref{th1}, and  is omitted here for brevity. }
\end{proof}
\begin{remark}
	Different from the static event-triggered average consensus problem in the previous sections, due to the introducing of the exogenous reference signals, the sampling errors  in the DAC problem are no longer images of linear operators acting on the consensus error of the sampled states.   Nevertheless, the triggering function ensures that these operators are finite-gain $\mathcal L_2$ stable and thus can be handled together with the additive dynamic uncertainties using the operator-theoretic approach.
\end{remark}

{ Define the DAC tracking error of the agent $i$ as $e_i=x_i-\frac{1}{N}\sum_{i=1}^Nr_i.$ Denoting  further  $p_i=x_i-r_i$,  $p=[p_1,\cdots,p_N]^T$, $e=[e_1,\cdots,e_N]^T$, $d=[d_1,\cdots,d_N]^T$, $r=[r_1,\cdots,r_N]^T$, $u=[u_1,\cdots,u_N]^T$ and $w=[w_1,\cdots,w_N]^T$,  the DAC system can then be rewritten in a compact form as follows:}
\begin{equation}\label{eq23}
	\begin{aligned}
		&\begin{bmatrix}
			\dot p \\ \dot w
		\end{bmatrix}=\begin{bmatrix}
			0 & -\beta L\\
			\theta I & -\theta I-\beta L
		\end{bmatrix}\begin{bmatrix}
			p\\ w
		\end{bmatrix}\\
		&\qquad+ \begin{bmatrix}
			0 & -\beta L& 0\\
			\theta I& -\beta L& \theta I
		\end{bmatrix}\begin{bmatrix}
			d\\\epsilon\\r 
		\end{bmatrix},\\
		& \begin{bmatrix}
			u\\ z\\e
		\end{bmatrix}=
		\begin{bmatrix}
			0 &-\beta L\\
			0 & M\\
			I & 0\\
		\end{bmatrix}\begin{bmatrix}
			p\\ w
		\end{bmatrix}+
		\begin{bmatrix}
			0 &-\beta L & 0\\
			0 & M & 0\\
			0 & 0 & M
		\end{bmatrix}\begin{bmatrix}
			d\\\epsilon\\r 
		\end{bmatrix}.
	\end{aligned}
\end{equation}
Denoting $\tilde u =U^Tu$, $\bar u = \tilde u_{2:N}$, $\tilde z =U^Tz$, $\bar z = \tilde z_{2:N}$, $\tilde d =U^Td$, $\bar d = \tilde d_{2:N}$, $\tilde \epsilon =U^T\epsilon$, $\bar \epsilon = \tilde \epsilon_{2:N}$, $\tilde p =U^Tp$, $\bar p = \tilde p_{2:N}$, $\tilde w =U^Tw$, $\bar w = \tilde w_{2:N}$, $\tilde e=U^Te$, $\tilde r=U^Tr$, we  have 
\begin{equation}\label{eq24}
	\begin{aligned}
		&\begin{bmatrix}
			\dot {\bar p} \\ \dot {\bar w}
		\end{bmatrix}=\begin{bmatrix}
			0 & -\beta \bar \Lambda \\
			\theta I & -\theta I-\beta \bar \Lambda
		\end{bmatrix}\begin{bmatrix}
			 \bar p\\ \bar w
		\end{bmatrix}\\
		&\qquad\quad+ \begin{bmatrix}
			0 & -\beta \bar \Lambda& 0\\
			\theta I& -\beta \bar \Lambda& \theta I
		\end{bmatrix}\begin{bmatrix}
			\bar d\\\bar\epsilon\\\tilde r 
		\end{bmatrix},\\
		& \begin{bmatrix}
			\bar u\\ \bar z\\\tilde e
		\end{bmatrix}=
		\begin{bmatrix}
			0 &-\beta \bar \Lambda\\
			0 & I\\
			I & 0\\
		\end{bmatrix}\begin{bmatrix}
			\bar p\\ \bar w
		\end{bmatrix}+
		\begin{bmatrix}
			0 &-\beta \bar \Lambda & 0\\
			0 & I & 0\\
			0 & 0 & M
		\end{bmatrix}\begin{bmatrix}
			\bar d\\\bar\epsilon\\\tilde r 
		\end{bmatrix}.
	\end{aligned}
\end{equation}
Note  that 
$$\begin{aligned}
& \bar d=Y^T \Delta^a(Y\bar u)= \bar \Delta^a (\bar u),	\\
& \bar \epsilon = Y^T\Delta^e(Y\bar z) =\bar\Delta^e(\bar z),
\end{aligned}
$$
where $\bar\Delta^a,\bar \Delta ^b:\mathcal L_2\mapsto\mathcal L_2  
$ are nonlinear  finite-gain operators with  $\|\bar\Delta^a\|_{\infty}\leq\eta$ and $\|\bar\Delta^e\|_{\infty}\leq \sqrt{\alpha \lambda_N}$.

In this section, we aim to examine the robustness of the event-triggered DAC algorithm \dref{eq20} and \dref{eq21} under additive dynamic uncertainties, which lies in two aspects: First, when there is no external reference signal $r$, the states $x$ and $w$ in the DAC algorithm achieve robust consensus; Second, when there is a reference signal $r$, we want to see how small the $H_{\infty}$ norm of the transfer function from the reference signal $r$ to the average  tracking error $e$, i.e., $\|T_{re}\|_{\infty}$, could be.

We now address the first problem and assume that the reference signal $r$ does not exist temporarily. It is obvious that $x$ and $w$ reach consensus if and only if $\bar p$ and $\bar w$ are  both asymptotically stable. It is easy to derive from \dref{eq24} that 
\begin{equation}\label{eq25}
	\begin{aligned}
	&\begin{bmatrix}
	\bar d\\\bar \epsilon
	\end{bmatrix}=\begin{bmatrix}
		\bar \Delta^a &0\\ 0 & \bar \Delta^e
	\end{bmatrix}\begin{bmatrix}
		\bar u\\\bar z 
	\end{bmatrix},\\
		&\begin{bmatrix}
			\bar u\\\bar z 
		\end{bmatrix}=\begin{bmatrix}
			G_{11} & G_{12}\\
			G_{21} & G_{22}
		\end{bmatrix}
		\begin{bmatrix}
			\bar d\\ \bar\epsilon 
		\end{bmatrix},
	\end{aligned}
\end{equation}
where $$\begin{aligned}
G_{11} &= -\theta\beta s \bar \Lambda\Pi, ~G_{12} = -\beta(s^2+s\theta ) \bar\Lambda \Pi,\\
G_{21} &= s\theta\Pi,~G_{22}= (s^2+s\theta )\Pi,\\
\Pi &= (s^2I+s(\theta I+\beta \bar \Lambda) +\theta\beta \bar \Lambda)^{-1}.
\end{aligned}$$

Following similar lines in proving Theorem \ref{th44}, we can obtain the following theorem which gives a sufficient condition for the event-triggered DAC system to reach robust consensus.
\begin{theorem}
	Let $\gamma=\max \{\eta, \sqrt{2\alpha \lambda_N}\}$. The uncertain network system \dref{eq20} reaches robust consensus  under the event-triggered DAC algorithm  \dref{eq21}, if $\frac{\theta\beta\lambda_N}{\theta+\beta\lambda_N}+1<\frac{1}{\gamma}$. 
\end{theorem}

\begin{remark}
	This theorem unveils the quantitative relationship among the robustness of the event-triggered DAC algorithm against  aperiodic event triggering and additive dynamic uncertainties, the parameters $\theta$ and $\beta$ in the DAC algorithm and the largest eigenvalue of the Laplacian matrix.  
\end{remark}

Next, we move on to consider the DAC tracking performance, quantified  by  $\|T_{re}\|_{\infty}$. Given the reference signal $r$, the smaller the $\|T_{re}\|_{\infty}$ is, the smaller the tracking error would be.
From \dref{eq24}, we have  \begin{equation}\label{eq26}
	\begin{aligned}
	&\begin{bmatrix}
	\bar d\\\bar \epsilon
	\end{bmatrix}=\begin{bmatrix}
		\bar \Delta^a &0\\ 0 & \bar \Delta^e
	\end{bmatrix}\begin{bmatrix}
		\bar u\\\bar z \\
	\end{bmatrix},\\
		&\begin{bmatrix}
			\bar u\\\bar z \\\tilde e
		\end{bmatrix}=\begin{bmatrix}
		G_{11} & G_{12} & G_{13}\\
		G_{21} & G_{22} & G_{23}\\
		G_{31} & G_{32} & G_{33}
	\end{bmatrix}
		\begin{bmatrix}
			\bar d\\ \bar\epsilon\\\tilde r 
		\end{bmatrix},
	\end{aligned}
\end{equation}
where 
$$\begin{aligned}
 G_{13} &=\begin{bmatrix}
	0_{N-1} &-\beta s\theta \bar \Lambda\Pi  
\end{bmatrix},\\
 G_{23} &=\begin{bmatrix}
	0_{N-1} & s\theta \Pi  
\end{bmatrix} ,\\
 G_{31} &=\begin{bmatrix}
	0_{N-1} & -\Pi\theta\beta \bar\Lambda
\end{bmatrix}^T,\\
 G_{32} &= \begin{bmatrix}
	0_{N-1} & -\Pi (s\beta \bar\Lambda+\theta\beta \bar\Lambda ) 
\end{bmatrix},\\
 G_{33} &=\begin{bmatrix}
	M-\theta\beta\tilde\Pi  \Lambda
\end{bmatrix},\\
\tilde\Pi &=(s^2I+s(\theta I+\beta \Lambda) +\theta\beta \Lambda)^{-1}.
\end{aligned}
$$
Note that \dref{eq26} is an upper linear fractional transformation form and the closed-loop transfer function is $T_{\tilde r\tilde e}=\mathcal F_u(G,\Delta)$. Therefore, the performance index $\|T_{re}\|_{\infty}=\|T_{\tilde r\tilde e}\|_{\infty}=\|\mathcal F_u(G,\Delta)\|_{\infty}$.

\begin{theorem}\label{th8}
	Suppose that the parameter  $\alpha$, $\theta$, $\beta$, $\mu$ and $\nu$ are  chosen such that \begin{equation}\label{condition}
		(\frac{\theta}{\beta\lambda_2}+1)^2+(\frac{\theta\beta\lambda_N}{\theta+\beta\lambda_N}+1)^2+2\frac{\theta\beta\lambda_N}{\theta+\beta\lambda_N}+\frac{2}{\beta\lambda_2}<\frac{1}{\gamma^2},
	\end{equation}
	where $\gamma=\max\{\eta,\sqrt{2\alpha\lambda_N}\}$.
	Then the DAC tracking performance $\|T_{re}\|_{\infty}\leq\frac{1}{\gamma}$ is achieved by the event-triggered protocol \dref{eq21}. Moreover, the closed-loop system does not exhibit the Zeno behavior. 
\end{theorem}
\begin{proof}
	In light of Lemma \ref{l5}, to achieve the robust tracking performance $\|T_{re}\|_{\infty}<\frac{1}{\gamma}$ for all $\left[\begin{smallmatrix}
		\bar\Delta^a & \\
		& \bar \Delta^b
	\end{smallmatrix}\right ] \in\left[\begin{smallmatrix}
		\mathbf \Delta_1 & \\
		& \mathbf \Delta_2
	\end{smallmatrix}\right ]  $ and $\left\|\left[\begin{smallmatrix}
		\bar\Delta^a & \\
		& \bar \Delta^b
	\end{smallmatrix}\right ]\right \|_{\infty}\leq \gamma,$ it is enough to let $$\mu_{\tilde{\mathbf \Delta}}\left(\begin{bmatrix}
		G_{11} & G_{12} & G_{13}\\
		G_{21} & G_{22} & G_{23}\\
		G_{31} & G_{32} & G_{33}
	\end{bmatrix}\right)<\frac{1}{\gamma},$$ 
	where $$\tilde{ \mathbf\Delta}=\begin{bmatrix}
		\mathbf \Delta_1 & &\\
		& \mathbf \Delta_2 &\\
		& & \mathbf \Delta_f
	\end{bmatrix}.$$
	 	It can be derived by some simple calculations that 
	{ $$\begin{aligned}
	&\|G_{11}(j\omega)\|^2 = \frac{\omega^2\theta^2\beta^2\lambda_N^2}{(\omega^2+\theta^2)(\omega^2+\beta^2\lambda_N^2)},\\
	&\|G_{12}(j\omega)\|=\sqrt{\frac{\beta^2\lambda_N^2\omega^2}{\omega^2+\beta^2\lambda_N^2}},\\
	&\|G_{21}(j\omega)\|
	 = \sqrt{\frac{\omega^2\theta^2}{(\omega^2+\theta^2)(\omega^2+\beta^2\lambda_2^2)}},\\
	&\|G_{22}(j\omega)\|^2
	 = \frac{\omega^2}{\omega^2+\beta^2\lambda_2^2},\\
		&\|G_{13}(j\omega)\|=\sqrt{\frac{\theta^2\beta^2\lambda_N^2\omega^2}{(\omega^2+\theta^2)(\omega^2+\beta^2\lambda_N^2)}}\leq \frac{\theta\beta\lambda_N}{\theta+\beta\lambda_N},\\
		& \|G_{31}(j\omega)\|=\sqrt{\frac{\theta^2\beta^2\lambda_N^2}{(\omega^2+\theta^2)(\omega^2+\beta^2\lambda_N^2)}}\leq 1,\\
		& \|G_{23}(j\omega)\|= \sqrt{\frac{\theta^2\omega^2}{(\omega^2+\theta^2)(\omega^2+\beta^2\lambda_2^2)}}\leq 1,\\
		& \|G_{32}(j\omega)\|=\sqrt{\frac{\beta^2\lambda_N^2}{\omega^2+\beta^2\lambda_N^2}}\leq 1,\\
		& \|G_{33}(j\omega)\|^2=\max_{i=2,\cdots,N}\frac{\omega^4+(\theta+\beta\lambda_i)^2\omega^2}{(\omega^2+\theta^2)(\omega^2+\beta^2\lambda_i^2)} \\
		& \quad\qquad\qquad\leq (\frac{\theta}{\beta\lambda_2}+1)^2.
	\end{aligned}$$}
	In light of Lemma \ref{l4}, 
 we can derive that $\|T_{re}\|_{\infty}<\frac{1}{\gamma}$, if \dref{condition} holds. The Zeno behavior can be similarly excluded as in Theorem \ref{th2} and  is omitted here for brevity.  
\end{proof}
\begin{remark}
{  Not surprisingly, we can easily see from \dref{condition} that a necessary condition of the robust performance specification is $\frac{\theta\beta\lambda_N}{\theta+\beta\lambda_N}+1<\frac{1}{\gamma}$. This coincides with the fact  that} {  the  robust performance specification  requires that the  robust consensus specification holds when there are no reference signals. Moreover, it can also be observed from \dref{condition} that to ensure that the event-triggered DAC algorithm has a better robustness and tracking performance under additive uncertainties,  the proportional gain $\theta$ should be chosen to be as small as possible. Moreover, the parameter $\alpha$ in the event-triggered mechanism should be  relatively small. This is because when $\alpha$ increases, the gain of the operator will be larger, which will lead to  larger sampling errors and thus a worse control performance. } 
	 
\end{remark}

\section{Simulation Results }\label{s6}
{\color{black}In this section, simulation examples will be provided to validate the effectiveness  of the theoretical results.  

Consider a multi-agent system with six agents whose communication graph is shown in Fig. \ref{graph}.}
\begin{figure}[!t]
\centerline{\includegraphics[width=0.7\columnwidth]{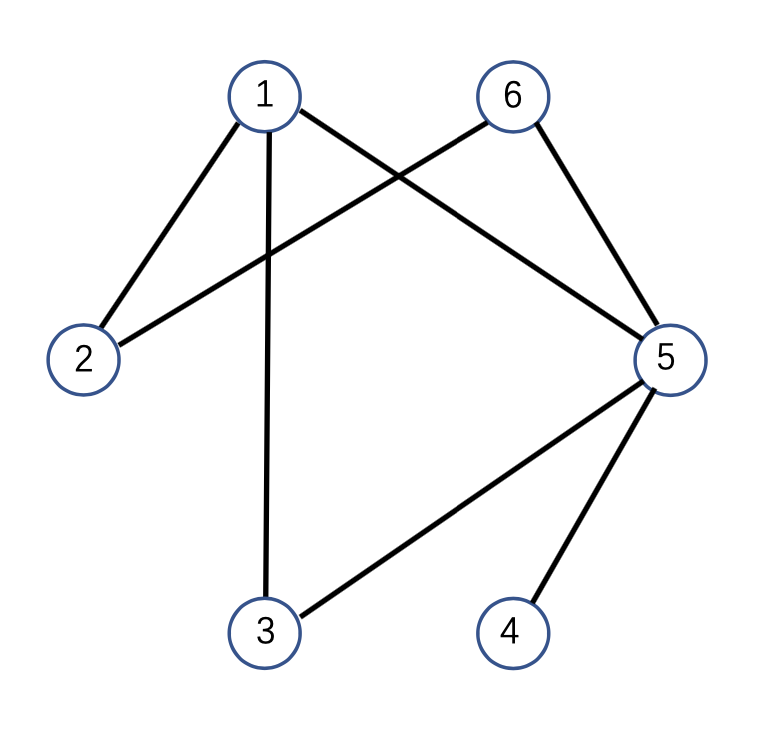}}
\caption{The network graph.}
\label{graph}
\end{figure}
{ The Laplacian matrix of the graph is $$L=\begin{bmatrix}
	3.5&-1&-2& 0 &-0.5& 0\\
    -1& 2&0 &0 &0 &-1\\
    -2& 0 & 3& 0& -1& 0\\
    0&0& 0& 1&-1& 0\\
    -0.5&0& -1&-1&3.5&-1\\
     0&-1& 0&0&-1& 2
\end{bmatrix}.$$ 
We randomly generate  transfer matrices representing the additive dynamic uncertainties from $\Delta^a_1$ to $\Delta^a_6$ as follows: 
$$\Delta^a_1=\left[\begin{smallmatrix} 
\begin{array}{cc|c}
	-55.4& 140.7 & -8.24\\  -155.7 &-71.41 &- 1.28\\ \hline 
	3.3989 & -5.4689 & 0.1022
\end{array}
\end{smallmatrix}\right],$$ 
$$\Delta^a_2=\left[\begin{smallmatrix} 
\begin{array}{cc|c}
	-55.4& 140.7 & -5.1520\\  -163.7 &-39.56 & -9.2230\\ \hline 
	2.03 & -2.14 & 0.2760
\end{array}
\end{smallmatrix}\right],$$ 
$$\Delta^a_3=\left[\begin{smallmatrix} 
\begin{array}{c|c}
	-0.3 & -0.28\\   \hline 
	0.5454  & 0.0460
\end{array}
\end{smallmatrix}\right],$$ 
$$\Delta^a_4=\left[\begin{smallmatrix} 
\begin{array}{cc|c}
	-5.4& 14.7 & -1.24\\  -15.7 &-1.41 & -0.28\\ \hline 
	0.1150 & -2.4610 & 0.0920
\end{array}
\end{smallmatrix}\right],$$ 
$$\Delta^a_5=\left[\begin{smallmatrix} 
\begin{array}{cc|c}
	-55.4& 140.7 & -3.0150\\  -155.7 &-71.41 & -3.1122 \\ \hline 
	0.33 & -2.14 & 0.4133
\end{array}
\end{smallmatrix}\right],$$ 
$$\Delta^a_6=\left[\begin{smallmatrix} 
\begin{array}{ccc|c}
	-44.4& 140.7 & -57.4 & -0.24\\  -19.7 &-18.41 & -6.32 & -1.28\\  45.70 & 29 & -130.84 & 1.16 \\ \hline 
	4.4563 & -10.2542 & 4.2646	 & 0.2875
\end{array}
\end{smallmatrix}\right].$$

Note that $\Delta_i^a\in \mathcal{RH}_{\infty}$ $\forall i=1,\cdots,6$ and $\|\Delta_i^a\|_{\infty}\leq0.4654=\eta$.  We then choose $\alpha=0.02$, $\mu=0.1$ and $\nu=5$ in the triggering function \dref{triggerfunction2}. It is easy to calculate that $\gamma=\max\{\eta,\sqrt{2\alpha\lambda_6}\}=0.4680$.  We select the controller gain $\beta$ to be $0.2$ such that the robust consensus condition in Theorem \ref{th44} is satisfied. Theorem \ref{th44} states that the network reaches consensus asymptotically under the event-triggered algorithm \dref{eq2}. To illustrate this, we depict the evolution of the state $x$   in Fig. \ref{consensus}. {  Denote the consensus error $\zeta=Mx$. The time instant after which $\|\zeta\|<0.1$ is denoted as $t_{min}$. When $\beta=0.2$, $t_{min}=14.7304$.}
\begin{figure}[!t]
\centerline{\includegraphics[width=\columnwidth,height=0.6\linewidth]{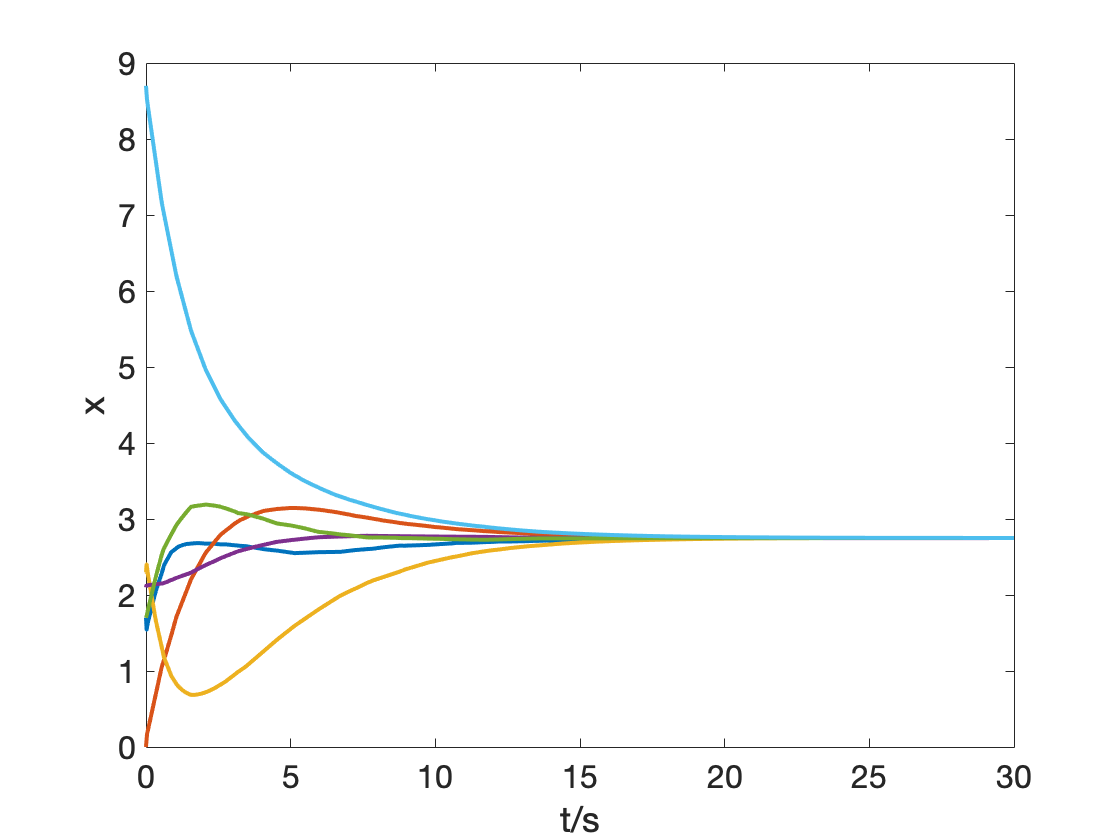}}
\caption{The evolution of the state $x$ when $\beta=0.2$ under additive dynamic uncertainties.}
\label{consensus}
\end{figure}
We then decrease the parameter $\beta$ to $0.1$ and let the other parameters remain unchanged. The evolution of the state $x$ with respect to the time $t$ is depicted in Fig. \ref{slowconsensus}. {  When $\beta=0.1$, $t_{min}=26.6796$, which is evidently larger than that when $\beta=0.2$.   Clearly, the  convergence  speed slows down with the decrease of $\beta$.}
\begin{figure}[!t]
\centerline{\includegraphics[width=\columnwidth,height=0.6\linewidth]{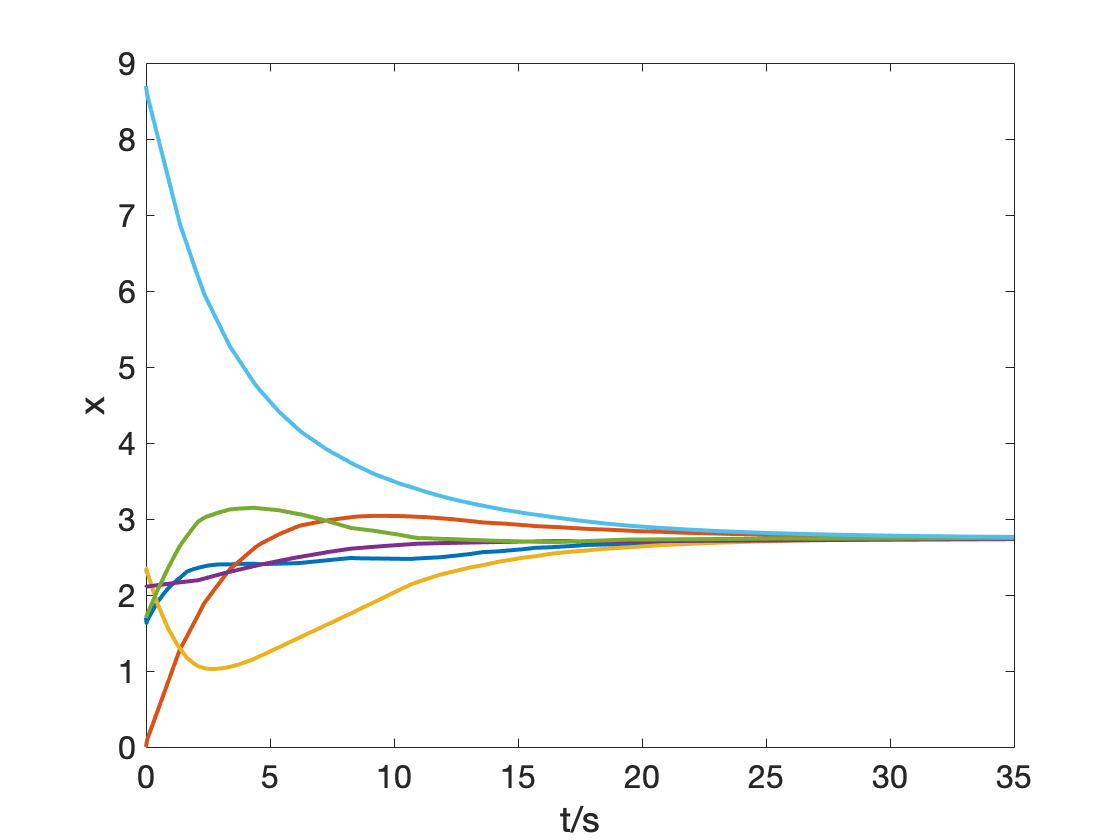}}
\caption{The evolution of the state $x$ when $\beta=0.1$ under additive dynamic uncertainties.}
\label{slowconsensus}
\end{figure}
On the other hand, if we increase the parameter $\beta$ for too much, then in virtue of Theorem \ref{th44}, the robust consensus cannot be guaranteed. For example, if we increase $\beta$ to $1.2$. Then the state $x$ cannot reach consensus anymore.
To show that the closed-loop system does not exhibit the Zeno behavior, we draw the triggering instants of the six agents when $\beta=0.2$ in Fig. \ref{triggertime}.
\begin{figure}[!t]
\centerline{\includegraphics[width=\columnwidth,height=0.6\columnwidth]{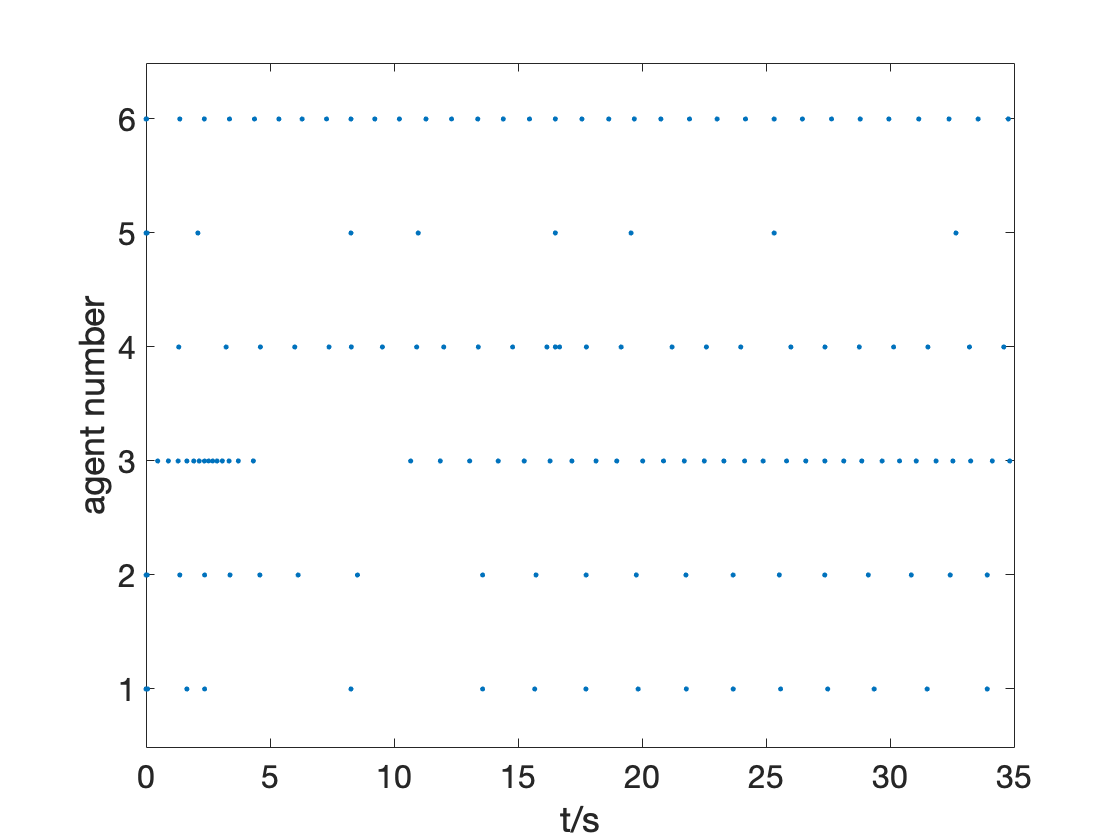}}
\caption{Triggering  instants of the six agents when $\beta=0.1$ under additive dynamic uncertainties.}
\label{triggertime}
\end{figure}

{ We next move on to test the robustness against network topology uncertainties and verify the result in Theorem \ref{th5}. We consider  the  communication topology  as in Fig.\ref{graph}. For illustration, we randomly generate a dynamical perturbation for each edge within a $H_{\infty}$ norm bound $\delta=0.1315$, namely, $\Delta_i^c$,$i=1,\cdots,7$. To satisfy the condition provided in Theorem \ref{th5}, we choose $\beta=0.08$, $\alpha=0.002$, $\mu=0.1$ and $\nu=5$. It is easy to verify that $\gamma=\max{\{\delta, \sqrt{2\alpha\lambda_6}\}}=0.1480$ and $\gamma<\sqrt{\frac{\lambda_2^2}{\lambda_2^2+\lambda_6^2}}.$ The evolution of the states of the six agents is depicted in Fig.\ref{consensus_network_uncertainty}. It is shown that the six agents reach state consensus. Furthermore, the triggering instants of the six agents are drawn in Fig.\ref{triggertime_network_uncertainty}.}

\begin{figure}[!t]
\centerline{\includegraphics[width=\columnwidth,height=0.6\linewidth]{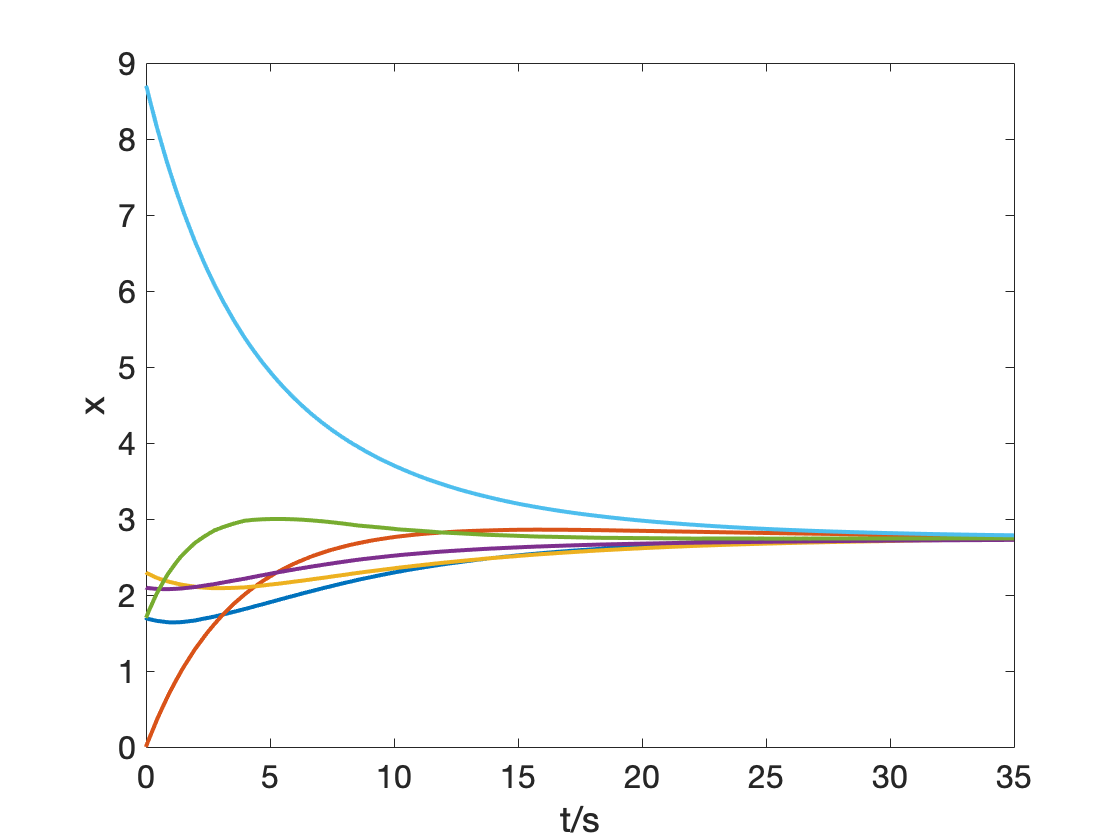}}
\caption{The evolution of the state $x$ under network topology uncertainties.}
\label{consensus_network_uncertainty}
\end{figure}

\begin{figure}[!t]
\centerline{\includegraphics[width=\columnwidth,height=0.6\columnwidth]{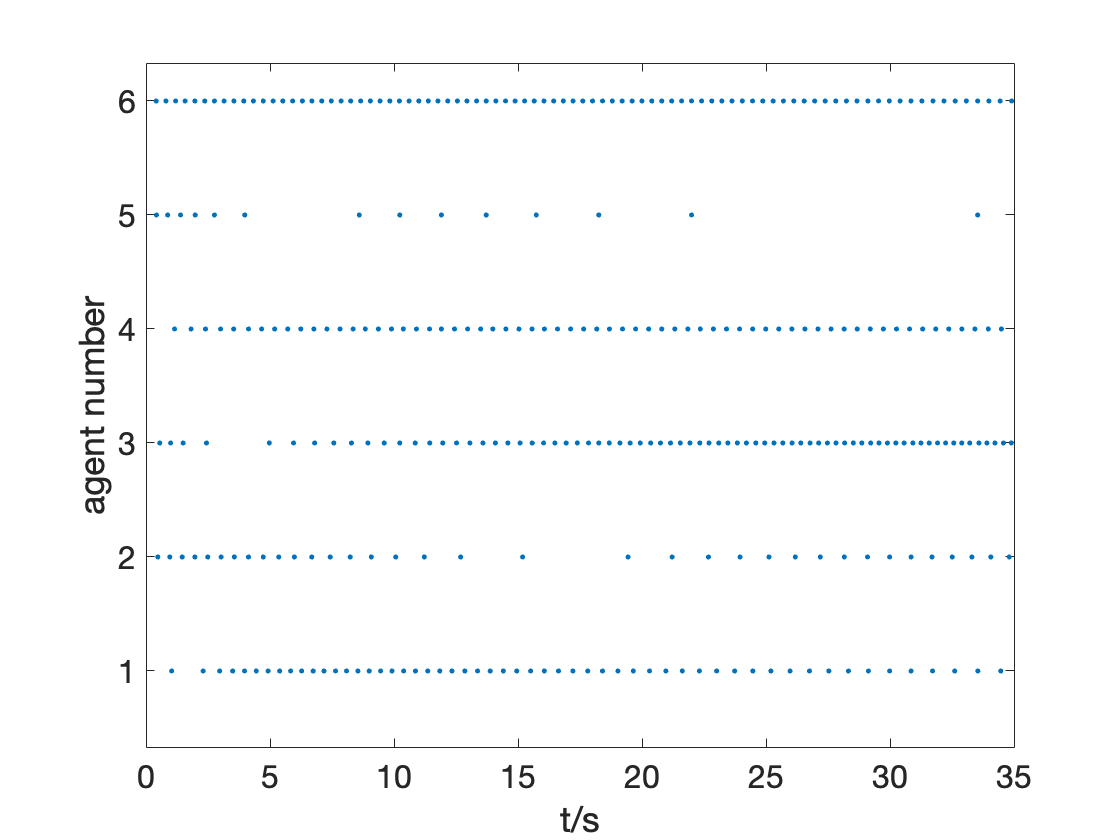}}
\caption{Triggering instants of the six agents under network topology uncertainties.}
\label{triggertime_network_uncertainty}
\end{figure}

\color{black}To test the robust performance of the event-triggered DAC algorithm \dref{eq21} in the presence of additive dynamic uncertainties, we choose $\alpha=0.02$, $\mu=0.1$, and $\nu=5$ in the triggering function \dref{eq22}. The perturbations $\Delta_i^a$ are the same as in the first case (the case when agents are subject to additive dynamic uncertainties). 
To satisfy the robust performance specification in Theorem \ref{th8}, we can choose $\theta=0.25$ and $\beta=1.2$ in the DAC algorithm \dref{eq21}. Supposing that the reference signals whose average  the agents need to track are $r=[6.1\sin(0.02t);19.1\cos(0.02t);4.8\cos(0.07t+8);2.2\sin(0.06t);1.9\cos(0.041t)e^{-0.09t};2.5\sin(0.05t)]^T$, we can calculate the evolution of the average tracking error $e$. By depicting the evolution of $e$, Fig. \ref{trackingerror} illustrates that the event-triggered DAC  algorithm has a pre-specified tracking performance. 
\begin{figure}[!t]
\centerline{\includegraphics[width=\columnwidth,height=0.6\columnwidth]{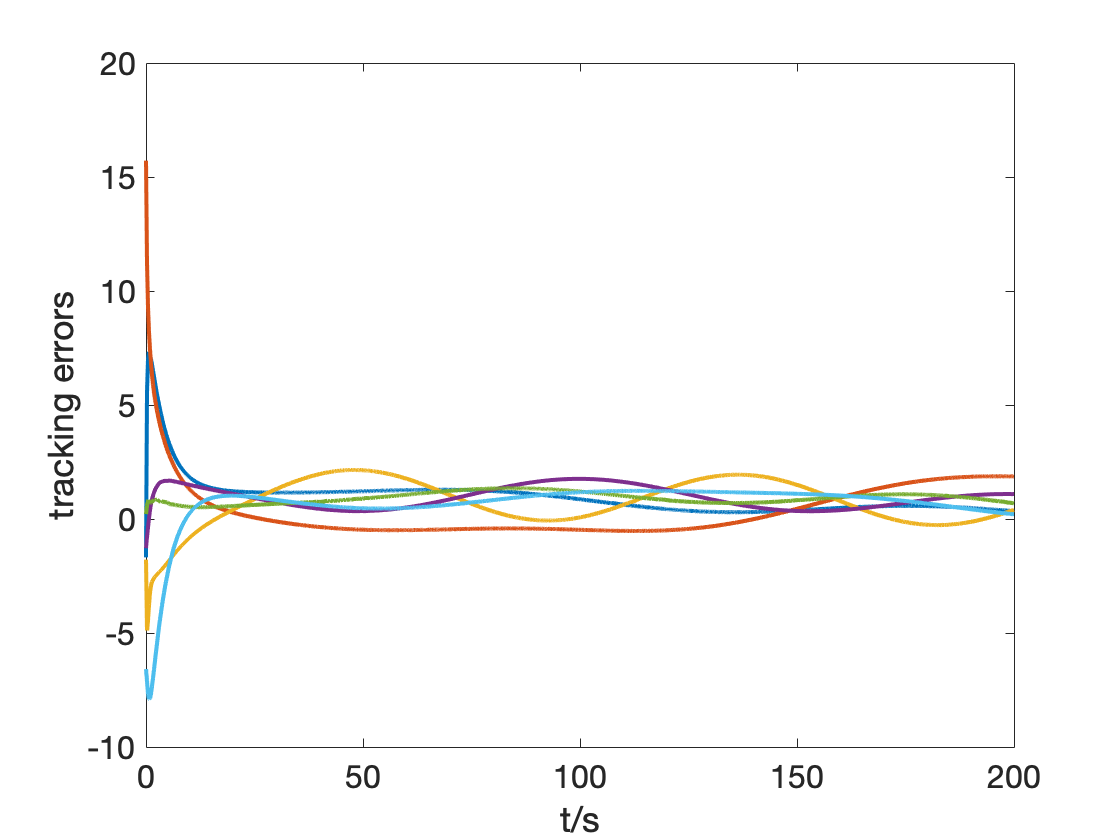}}
\caption{The evolution of  the DAC tracking errors $e$.}
\label{trackingerror}
\end{figure}
The triggering instants of the six agents from $194s$ to $200s$ are depicted in  Fig. \ref{triggertime2}.

\begin{figure}[!t]
\centerline{\includegraphics[width=\columnwidth,height=0.6\columnwidth]{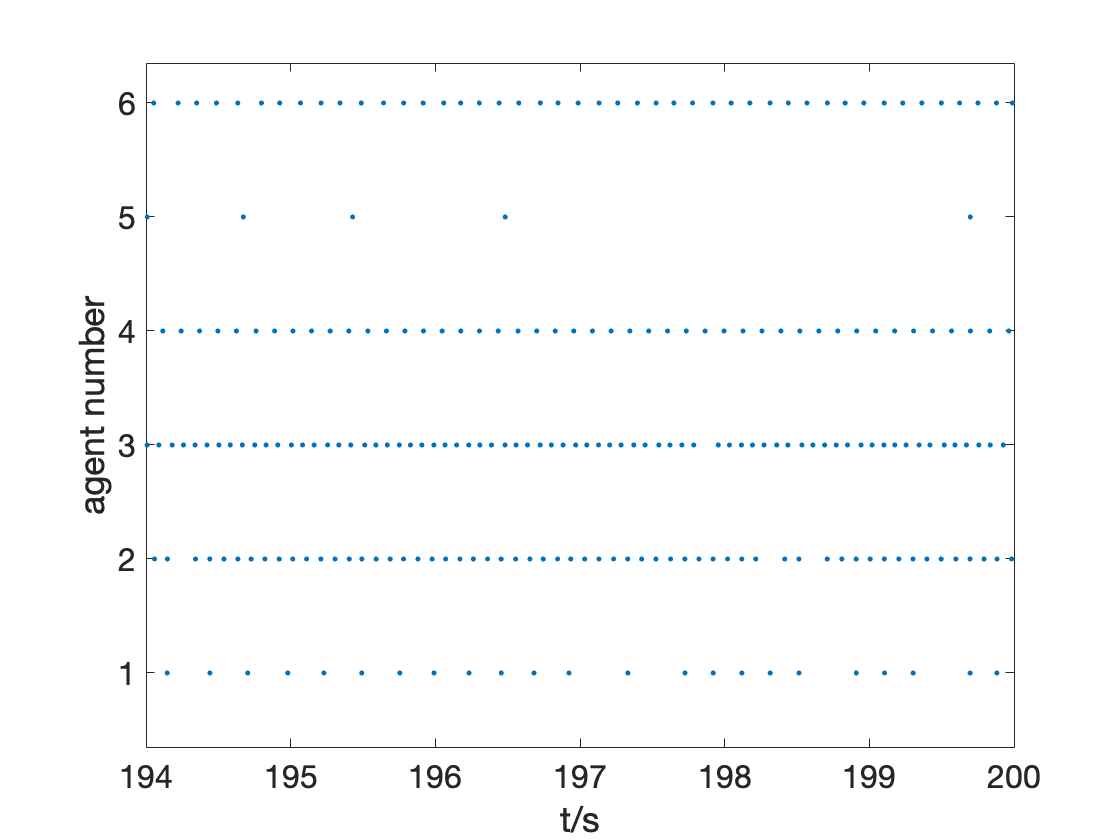}}
\caption{Triggering instants of the agents from $194s$ to $200s$.}
\label{triggertime2}
\end{figure}

\section{Conclusion}\label{conclusion}
In this paper, a novel operator-theoretic approach has been put forward to study the robustness of the event-triggered consensus algorithms against frequency-domain uncertainties. By treating the event-triggered sampling mechanism as a negative feedback loop and the sampling errors resulted by event triggering as  the images of the finite-gain $\mathcal L_2$ stable operators, a frequency-domain analysis framework has been established. 
The developed approach  effectively  builds a bridge between the time-domain triggering mechanism and the frequency-domain uncertainties.

There are many potential extensions to this paper. Perhaps the most direct one is to consider the case where the agent dynamics are   higher-order integrators or even  general linear systems. It should be noted that the introduction of the system matrix $A$ will definitely make the problem much more challenging. This is  for sure an important problem we will consider in the future. }

\end{document}